\documentclass{article}
\usepackage{arxiv}
\usepackage{amsmath, amssymb, amsthm}
\newtheorem{lemma}{Lemma}
\newtheorem{theorem}{Theorem}
\newtheorem{corollary}{Corollary}
\newtheorem{definition}{Definition}

\usepackage[numbers]{natbib}

\usepackage[T1]{fontenc}
\usepackage[utf8]{inputenc}
\usepackage{amsmath}
\usepackage{amssymb}
\usepackage{algorithm}
\usepackage{algorithmic}
\usepackage{float}
\usepackage{booktabs}
\usepackage{hyperref}
\usepackage{listings}
\usepackage{tikz}
\usetikzlibrary{positioning}
\usetikzlibrary{arrows.meta}
\usepackage{lineno}
\usepackage{inconsolata}
\usepackage{enumitem}
\usepackage{xcolor}
\usepackage{placeins}

\newcommand{\expll}[1]{\mathrm{expl}(#1)}

\DeclareMathOperator{\bsdfs}{\textsc{BS-DFS}}
\DeclareMathOperator{\search}{\textsc{Search}}
\DeclareMathOperator{\fruitful}{\textsc{Fruitful}}
\DeclareMathOperator{\ER}{\mathrm{ER}}
\DeclareMathOperator{\WS}{\mathrm{WS}}
\DeclareMathOperator{\pre}{\mathrm{pred}}
\DeclareMathOperator{\suc}{\mathrm{succ}}

\newcommand{\fullref}[1]{\autoref{#1} (\nameref{#1})}

\renewcommand{\shorttitle}{Enumerating Length-Bounded Simple Paths and Cycles}
\renewcommand{\undertitle}{}
\renewcommand{\headeright}{}

\title{Enumerating Length-Bounded Simple Paths and Cycles in Directed Graphs with $O(k(n+m))$ 
    Delay Using Edge-Consistent Node Barriers}
\author{
  Frank Bauernöppel \\
  Computer Engineering\\
  Hochschule für Technik und Wirtschaft, Berlin, Germany\\
  \texttt{frank.bauernoeppel@htw-berlin.de}
  \And
  Jörg-Rüdiger Sack \\
  Carleton University Ottawa, School of Computer Science, Canada \\
  \texttt{sack@scs.carleton.ca}
}
\date{}

\begin{document}






\maketitle


\begin{abstract}
Enumerating simple paths and cycles subject to a given length bound is a
fundamental problem in graph algorithms with applications
from network analysis to fraud detection.
Recent algorithms for this problem, namely
$\textsc{BC-DFS}$ (Peng et al., 2019, 2021) and
$\textsc{CYCLE\_SEARCH}$ (Gupta and Suzumura, 2021),
employ cached barrier values to prune fruitless searches.

Both algorithms turn out to produce incomplete output,
and their delay-bound arguments rely on flawed claims.
For $\textsc{CYCLE\_SEARCH}$ this is known;
here we establish the analogous results for $\textsc{BC-DFS}$ by
exhibiting graphs on which paths are missed, by identifying the defect
in its barrier-update procedure, and by refuting the monotonicity
claim on which its delay-bound proof rests.

As our main contribution, we introduce \emph{edge-consistency}, a local
invariant on barrier values analogous to heuristic consistency in informed
search. It provides an incremental mechanism for maintaining admissible
barrier estimates and yields concise correctness proofs.

We use edge-consistency as a unifying framework for design and analysis of
\emph{Bounded-Scope Depth-First Search} ($\textsc{BS-DFS}$) ---
a new algorithm for enumerating simple paths or cycles of length at most
$k$ in a directed graph. For $\textsc{BS-DFS}$
we prove a worst-case delay of at most $3(k+1)(n+m)$ elementary steps
between consecutive events (start, each output, termination) 
and an amortized delay of at most $2(k+1)(n+m)$ steps per event,
the $p$-th event being reached within $2p(k+1)(n+m)$ steps;
both bounds are in $O(k(n+m))$.
Barrier admissibility alone is not sufficient for the delay bound:
for two variants with simpler barrier management,
we exhibit a graph family forcing $\Omega(k^2(n+m))$ delay between outputs.

Experiments on two families of random graphs confirm our findings,
support the significance of the incompleteness result,
and show that achieving completeness has modest empirical cost.
\end{abstract}

\section{Introduction}

\subsection{Related Work}

The enumeration of simple paths and cycles in directed graphs is a
classical problem in graph algorithms, with applications ranging from
chemistry and circuit analysis to network science and computational biology.

The classical line of work places no bound on the length of the reported
paths or cycles. Early work by Tiernan~\cite{tiernan1970} introduced a
backtracking algorithm for enumerating elementary circuits, but suffered
from exponential run-time in some dense graphs;
Tarjan~\cite{DBLP:journals/siamcomp/Tarjan73} reduced this redundancy by
a marking scheme. A breakthrough was achieved by
Johnson~\cite{DBLP:journals/siamcomp/Johnson75}, whose influential
algorithm combines depth-first search with a Boolean blocking-unblocking
mechanism to avoid repeatedly exploring fruitless search branches;
it runs in $O((c+1)(n+m))$ time, where $c$ denotes the number of simple
cycles. Subsequent work, including that of Szwarcfiter and
Lauer~\cite{szwarcfiter1976}, refined these ideas, and Tarjan's algorithm
for strongly connected components~\cite{DBLP:journals/siamcomp/Tarjan72}
continues to serve as a preprocessing step in many enumeration
algorithms. Birmelé et al.~\cite{birmele2013} give output-sensitive
algorithms that list \emph{all} simple cycles, and all simple
$s$-$t$-paths, in \emph{undirected} graphs, optimal in the total output
size. A survey on enumeration algorithms in graphs was compiled in 2016
by Grossi~\cite{Grossi2016}.

Since the number of simple paths and cycles may grow exponentially with
the size of the graph, unrestricted enumeration is generally impractical
on large instances; moreover, many associated decision problems,
including the classical Hamiltonian path and Hamiltonian cycle problems,
are NP-complete~\cite{GareyJohnson1979}. Practical applications therefore
often restrict the problem, most commonly by a bound $k$ on the path or
cycle length. In many contemporary applications the underlying graphs are
huge, but only paths or cycles up to a small length bound are of
interest; examples are the analysis of large social networks, fraud
detection in transaction networks, and motif discovery in biological
networks.

Rizzi et al.~\cite{DBLP:journals/corr/RizziSS14} impose such a length
bound, listing all length-bounded simple $s$-$t$-paths. Their focus is
the \emph{weighted} case, matching the classic $K$-shortest-paths time bounds in only $O(n+m)$ space while
a short remark deals with the unweighted case. For the \emph{unweighted} case, they suggest
using a reverse (from the target) breadth-first search at each node visited
to eliminate the fruitless (no output generating) successors from the depth-first recursion.

The length-bounded setting has further stimulated the development of algorithms pruning
the depth-first search efficiently by maintaining per-node barrier values.
Two notable such algorithms are \textsc{CYCLE\_SEARCH}, introduced by Gupta and
Suzumura~\cite{DBLP:journals/corr/abs-2105-10094}, and
\textsc{BC-DFS}, proposed by Peng et al.~\cite{10.14778/3372716.3372720,peng_efficient_2021}.
Both algorithms build upon ideas originating in Johnson's algorithm,
but replace the classical binary blocking mechanism by numeric barrier
values. These barrier values are intended to cache information about
previously explored searches and thereby avoid redundant fruitless exploration.
These algorithms are shown to be very efficient in practical experiments,
and they are claimed to achieve a worst-case delay of $O((k-1)(n+m))$ and $O(km)$,
respectively, per output.

Unfortunately, the correctness of these barrier-based approaches has
proven to be more delicate than originally anticipated.
For $\textsc{CYCLE\_SEARCH}$, counter-examples and proof gaps were identified by
Bauernöppel and Sack~\cite{boundedcycles2025}.
As a remedy, they propose an approach similar to
Rizzi's unweighted case: use depth-first search starting at $s$ and prune
it by exact barrier values, freshly computed by a reverse (from $t$)
depth-limited breadth-first search for each node visited in the
depth-first search.
In contrast to Rizzi, the underlying graph $G$ is not mutated during recursion.

Although the breadth-first search at every node visit does not
affect the asymptotic delay bound, it repeatedly recomputes essentially the same information.
Peng et al.~\cite{10.14778/3372716.3372720,peng_efficient_2021}
(where a similar approach for path enumeration is named T-DFS)
report that
``T-DFS demonstrates the worst performance on most of the graphs although
it is a polynomial delay algorithm with a nice theoretical guarantee''.
The central challenge is therefore developing an efficient
barrier based output-sensitive enumeration algorithm
with proven correctness and delay bound.

\subsection{Our Contributions}

We introduce \emph{Bounded-Scope Depth-First Search}
($\textsc{BS-DFS}$), a simple and efficient algorithm for enumerating
all length-bounded simple paths or simple cycles in a directed graph
for a given node pair $s, t$ or a single node $s$, respectively.
While inspired by earlier barrier-based approaches,
$\textsc{BS-DFS}$ warrants independent presentation and analysis through
the notion of \emph{edge-consistency}, a local invariant that forms the
basis of its correctness proof.

$\textsc{BS-DFS}$ is an output-sensitive algorithm. Counting the start of
the run, each output, and the termination as the \emph{events} of that run,
we prove that $\textsc{BS-DFS}$ executes at most $3(k+1)(n+m)$ elementary
steps between consecutive events (\autoref{thm:worst-case-delay}), and,
amortized, reaches the $p$-th event within $2p(k+1)(n+m)$ elementary steps
(\autoref{thm:amortized-delay}). Both bounds are in $O(k(n+m))$.

Our experiments on two families of random graphs confirm the delay bound with
substantial margin. The largest delay observed between consecutive events
is $0.44\,(k+1)(n+m)$ elementary steps,
below a sixth of the constant proven in \fullref{thm:worst-case-delay},
and it does not grow over four decades of run length (\autoref{fig:delay-bounds}).
Restoring completeness costs a factor of $1.6$ in running time per interval
compared to $\textsc{BC-DFS}$ (\autoref{fig:runtime}).

We also present and discuss edge-consistent variants of $\textsc{BS-DFS}$,
varying in how barrier values are maintained, including a lazy variant
using dependency lists similar to Johnson's algorithm and $\textsc{CYCLE\_SEARCH}$.
The variants separate correctness from efficiency. The coarse \emph{loose} and
\emph{lazy} schemes remain complete under a markedly simpler argument, yet we
exhibit a graph family on which both require $\Omega(k^2(n+m))$ steps between
two consecutive outputs (\fullref{thm:loose-breaker}), so neither attains
$O(k(n+m))$ delay.

For $\textsc{BC-DFS}$~\cite{10.14778/3372716.3372720,peng_efficient_2021},
we show that completeness fails, by a counter-example,
and identify the underlying defect in its barrier-update procedure.
We further locate the point at which their correctness argument breaks down.

These contributions are not only of theoretical interest.
Both $\textsc{BC-DFS}$ and $\textsc{CYCLE\_SEARCH}$ have already been cited
in subsequent work~\cite{DBLP:conf/spaa/BlanusaIA22,DBLP:journals/topc/BlanusaAI23}
and $\textsc{CYCLE\_SEARCH}$ is used in the popular graph library
\textsc{NetworkX}~\cite{networkx,networkx_simple_cycles}.

Our experiments on random graphs show that the resulting omissions are
not rare edge cases: at length bound $k = 10$, $\textsc{BC-DFS}$ misses
$9.4\%$ of all paths on one graph family and $33.6\%$ on the other
(\autoref{fig:missed-paths}), in the latter case losing paths on every instance.
More broadly, the paper illustrates how subtle barrier propagation
schemes can be and why correctness arguments based on explicit
invariants are essential.

The paper is organized as follows. After fixing
terminology in \autoref{section:terminology} and presenting
$\bsdfs$ in \autoref{sec:algorithm}, we prove  soundness (\autoref{sec:soundness}),
establish completeness by introducing edge-consistency (\autoref{sec:completeness}),
and derive the delay bounds (\autoref{sec:delay-bounds}).
We then study several edge-consistent variants of $\bsdfs$ (\autoref{sec:bsdfs-variants})
and separate them by delay,
analyze the shortcomings of \textsc{BC-DFS} (\autoref{sec:analyzing-bcdfs}),
and report first experimental results (\autoref{sec:experiments})
before concluding in \autoref{sec:conclusions-and-outlook}.
All source code, counter-examples,
and experiments are available online at~\cite{bsdfs_github}.

\section{Terminology}\label{section:terminology}

A \emph{directed graph} $G$ is an ordered pair $G=(V,E)$,
where $V$ denotes a finite set of \emph{nodes} and $E \subseteq V \times V$
denotes a set of ordered pairs of nodes $(u,v)$, the \emph{edges} of $G$.
Let $n=|V|$ and $m=|E|$ denote the number of nodes and edges of $G$, respectively.
Throughout this paper, all graphs are directed and \emph{simple},
i.e., they contain neither self-loops ($(u,u) \notin E$) nor multiple edges.

For an edge $(u,v) \in E$, node $u$ is a \emph{predecessor} of $v$ and $v$ is a \emph{successor} of $u$.
The sets of predecessors and successors of a node $v$ are denoted by $\pre(v)$ and $\suc(v)$, respectively.
The edge $(u,v) \in E$ is said to be \emph{outgoing} from $u$ and \emph{incoming} to $v$.

A \emph{path} in $G$ is a sequence of nodes $P=(v_0,v_1,\ldots,v_l)$ such that
$(v_i,v_{i+1}) \in E$ for all $0 \le i < l$.
When a path $P$ begins at a node $x$ ($v_{0}=x$) and ends at a node $y$ ($v_{l}=y$),
it is called an $x$-$y$-path.
The \emph{length} of a path is the number of edges it contains, $\|P\| = l$.
A \emph{simple path} is a path in which all nodes are pairwise distinct.
A \emph{simple cycle} is a node sequence $C=(v_{0}, v_{1}, \dots, v_{l})$
with $l \ge 2$, $(v_{0}, v_{1}, \dots, v_{l-1})$ is a simple path, and $v_{l} = v_{0}$.
Two such sequences represent the same simple cycle if one is a cyclic rotation
of the other, i.e. $(v_0,\dots,v_l)$ and $(v_0',\dots,v_l')$ agree up to the
choice of starting node; a cycle is thus determined by its cyclic sequence of
edges, independently of which of its nodes is written first.
A simple cycle containing node $s$ is called a simple $s$-cycle.
The \emph{length} of a cycle is the number of edges it contains, $\|C\| = l$.
When appropriate, a path or cycle is identified with the set of nodes it contains.

For a path $P=(v_{0}, v_{1}, \dots, v_{l})$ and a node $v$ with $(v_l, v) \in E$, we write
$Pv$ for the path obtained by appending $v$ to $P$; iterating, $Pvw$ appends
first $v$, then $w$. The dot $\cdot$ is reserved for concatenating two paths,
as in $P_1 \cdot P_2$.

\begin{definition}[Relative Distance]
    Let $S \subseteq V$ be a set of nodes. For nodes $x, y$ we define
    \[
        dist_S(x,y) \;=\; \min \bigl\{ \, \|P\| \;:\;
            P \text{ is an } x\text{-}y\text{-path in } G
            \text{ with } P \cap S \subseteq \{x,y\} \, \bigr\} ,
    \]
    where $\min \emptyset = \infty$, so that $dist_S(x,y) = \infty$ if no such
    path exists.
    Thus $dist_S(x,y)$ is the length of a shortest path from $x$ to $y$
    among all paths which may intersect $S$ only at the endpoints $x$ and $y$.
    Note that $dist_\emptyset(x,y)$ is the ordinary edge distance
    $dist(x,y)$ in $G$.
    All distances are taken in the graph $G$ unless stated otherwise; where the ambient
    graph is ambiguous, we indicate it by a superscript, as in $dist^{G'}_S(x,y)$.
\end{definition}

The relative distance provides a triangle inequality in the following sense:

\begin{lemma}[Triangle Inequality]\label{lem:dist-triangle}
    Let $S \subseteq V$ be a set of nodes.
    For any nodes $x, z \in V$ and any intermediate node $y \in V \setminus S$,
    $$dist_S(x,z) \;\le\; dist_S(x,y) + dist_S(y,z).$$
\end{lemma}
\begin{proof}
    If either term on the right is $\infty$ the bound is trivial.
    Otherwise, let path $P_1$ realize $dist_S(x,y)$ and path $P_2$ realize $dist_S(y,z)$.
    The concatenation $P_1 \cdot P_2$ is an admissible $x$-$z$-path:
    its nodes meet $S$ only within $(\{x,y\} \cup \{y,z\}) \cap S = \{x,z\} \cap S$ because $y \notin S$.
    Hence, the minimum $dist_S(x,z)$ over all admissible paths satisfies
    $dist_S(x,z) \le \|P_1\| + \|P_2\| = dist_S(x,y) + dist_S(y,z)$.
\end{proof}

\section{Algorithm $\bsdfs$}\label{sec:algorithm}

\subsection{$\bsdfs$ for Simple $s$-$t$-Path Search}

Throughout, we assume $0 < k \le n$ and $s \ne t$ for finding all length $k$ bound
simple $s$-$t$-paths in the directed graph $G$. The upper bound $k \le n$ is
without loss of generality: a simple $s$-$t$-path has length at most $n - 1$;
the value $k = n$ is retained to cover simple cycle search, where a cycle through $s$ may have length up to $n$.

\begin{algorithm}[H]
    \caption{$\bsdfs$ --- Bounded-Scope Depth-First-Search}\label{alg:bsdfs}
    \textbf{\textsc{Global Data}}
    \begin{tabbing}
    \hspace{3em}\= \kill 
    \quad $G$:        \> directed graph \\
    \quad $s, t$:     \> start and target nodes \\
    \quad $k$:        \> length bound \\
    \quad $b[\cdot]$: \> barrier value for each node \\
    \quad $S$:        \> search-path / call stack
    \end{tabbing}
    \begin{algorithmic}[1]
        \STATE \textbf{$\bsdfs(G, s, t, k)$}
        \STATE \quad \textbf{for each} $v \in G.\mathrm{nodes}$ \textbf{do}
        \STATE \quad\quad $b[v] \gets 0$
        \STATE \quad $S \gets \langle\rangle$
        \STATE \quad $\search(s)$
    \end{algorithmic}
\end{algorithm}

\begin{algorithm}[H]
    \caption{$\search$}\label{alg:search}
    \begin{algorithmic}[1]
        \STATE \textbf{$\search(v)$}
        \STATE \quad append $v$ to $S$ \quad \# $S$ is now the search path
        \STATE \quad $h \gets \|S\|$
        \STATE \quad $sd \gets k + 1$
        \STATE \quad \textbf{for each} $w \in \suc(v)$ \textbf{do}
        \STATE \quad\quad \textbf{if} $b[w] + h < k$ \textbf{then}
        \STATE \quad\quad\quad \textbf{if} $w = t$ \textbf{then}
        \STATE \quad\quad\quad\quad output $S \cdot t$
        \STATE \quad\quad\quad\quad $sd \gets 1$
        \STATE \quad\quad\quad \textbf{else if} $w \notin S$ \textbf{then}
        \STATE \quad\quad\quad\quad $d \gets \search(w)$
        \STATE \quad\quad\quad\quad $sd \gets \min(sd, d + 1)$
        \STATE \quad \textbf{if} $sd \le k$ \textbf{then}
        \STATE \quad\quad $\fruitful(v, sd)$
        \STATE \quad \textbf{else}
        \STATE \quad \quad $b[v] \gets k - h + 1$ \quad \# fruitless
        \STATE \quad pop $v$ from $S$
        \STATE \quad \textbf{return} $sd$
    \end{algorithmic}
\end{algorithm}

\begin{algorithm}[H]
    \caption{$\fruitful$}\label{alg:fruitful}
    \begin{algorithmic}[1]
        \STATE \textbf{$\fruitful(v, sd)$}
        \STATE \quad $b[v] \gets sd$
        \STATE \quad $Q \gets$ queue containing tuple $(v, sd)$
        \STATE \quad \textbf{while} $Q$ not empty \textbf{do}
        \STATE \quad\quad $(q, d) \gets$ dequeue $Q$
        \STATE \quad\quad \textbf{for each} $p \in \pre(q)$ \textbf{do}
        \STATE \quad\quad\quad \textbf{if} $p \notin S$ \textbf{and} $b[p] > d + 1$ \textbf{then}
        \STATE \quad\quad\quad\quad $b[p] \gets d + 1$
        \STATE \quad\quad\quad\quad enqueue $(p, d + 1)$ onto $Q$
    \end{algorithmic}
\end{algorithm}

We give a short overview of the algorithm. All details will be elaborated upon
and proven in later sections. After initialization,
$\bsdfs(G, s, t, k)$ (\autoref{alg:bsdfs}) sets
the search-path stack $S$ empty and calls $\search(s)$, starting the recursive
depth-first exploration of $G$. A call $\search(v)$ (\autoref{alg:search})
appends its node $v$ to $S$ (line 2), so that during the call $S = (s = v_0,
v_1, \dots, v_h = v)$ is the current \emph{search path}, of length $h = \|S\|$
(line 3), which we also call the \emph{depth} of the call;
it mirrors the $\search$ call stack: a \emph{frame} --- the
activation record of a call $\search(v)$ --- is open exactly while
$v \in S$, the frame of $v_i$ having been pushed with search path
$(v_0, \dots, v_i)$, and the target test creates no frame for $t$.
We name a call by its current
node, $\search(v)$; in the prose we write $P$ for the contents of $S$ at entry
(the \emph{prefix}, before $v$ is appended), $Pv = P \cdot v$ for the search
path while the call is active, and $Pvw$ for the search path of the child call
on a successor $w$. These are snapshots of the single stack $S$ at successive
moments. The successors $w$ of $v$ are examined (line 5), each guarded by a
\emph{pruning condition} $b[w] + h < k$ (line 6) to avoid excessive fruitless
searches. When a successor equals the target ($w = t$), the path $S \cdot t$ is
output and $sd$ (``shortest distance'') is set to $1$ (lines 8--9); otherwise a
recursive search $\search(w)$ is performed at the successor $w$ (line 11).
During the for-loop, $sd$ tracks the shortest path length to $t$ found
via an edge to $t$ or in any child search (line 12). After the for-loop, $sd \le k$
marks the \emph{fruitful} case (line 13), in which output was produced, and $sd
= k + 1$ the \emph{fruitless} case (line 15).
In the fruitful case, $\fruitful(v, sd)$ (\autoref{alg:fruitful}) assigns $b[v] \gets sd$ and
propagates distance values $d$ backwards to predecessors in a breadth-first \emph{cascade};
in the fruitless case, $b[v] \gets k - h + 1$ (line 16).
Finally, $\search(v)$ pops $v$ from $S$ (line 17),
restoring the prefix, and returns $sd$ (line 18).

\paragraph{Note for implementations.}

The algorithm maintains the current search path in a single stack $S$,
appending a node upon entry to $\search$ and removing it immediately
before returning. By augmenting the stack with an auxiliary boolean
\emph{on-stack} flag for each node, both stack operations and the
membership tests --- $w \notin S$ in $\search$ and $p \notin S$ in
$\fruitful$ --- can be performed in $O(1)$ time.
The resulting space requirement remains $O(n+m)$.

\subsection{$\bsdfs$ for Simple $s$-Cycle Search}

The procedure $\search$ is formulated in an edge-oriented manner.
This leads to the following reduction.

\begin{lemma}[Cycle Reduction]\label{lem:cycle-reduction}
    Let $G'=(V',E')$ with $V' = (V\setminus\{s\}) \cup \{s_o, s_i\}$, where every
    edge $(s, v)$ of $G$ becomes $(s_o, v)$, every edge $(u, s)$ becomes $(u, s_i)$,
    and all edges not incident to $s$ are kept unchanged.

    Then $\bsdfs(G, s, s, k)$ and $\bsdfs(G', s_o, s_i, k)$ produce the same output.
    In particular, $\bsdfs(G, s, s, k)$ enumerates exactly all simple $s$-cycles of length $\le k$ in $G$.
\end{lemma}
\begin{proof}
The initial call $\search(s)$ explores only the \emph{outgoing} edges $(s,w) \in E$ of the source,
while the target test $w=t$ recognizes only the \emph{incoming} edges $(v,t)\in E$ of the target.
Thus, under the correspondence $s \leftrightarrow s_o$ and $t \leftrightarrow s_i$,
the two executions $\bsdfs(G, s, s, k)$ and $\bsdfs(G', s_o, s_i, k)$ coincide step by step
up to the last output.
The single barrier $b[s]$ of $G$ serves as both $b[s_o]$ and $b[s_i]$ of $G'$, but,
because of Observation~\ref{obs:source-barrier} and Observation~\ref{obs:target-barrier}, this is benign.
Thus, $\bsdfs(G, s, s, k)$ yields exactly the simple $s$-cycles of length $\le k$ in $G$.
\end{proof}

\subsection{Pre-processing for Finding All Simple Cycles}\label{sec:all-cycles}

Enumerating \emph{all} simple cycles of length at most $k$ in $G=(V,E)$ ---
not only those through a fixed node --- reduces to the single-node case by a
pre-processing step. Order the nodes as $V = \{v_1, v_2, \dots, v_n\}$ and let
$G^i$ denote the subgraph of $G$ induced by $V^i = \{v_i, v_{i+1}, \dots, v_n\}$.
Every simple cycle $C$ of length $\le k$ occurs in exactly one of
these graphs as a cycle through the distinguished node,
namely in $G^i$ for the smallest index $i$ among the nodes of $C$.

Preprocessing creates a list $L$ of all indices $i$
where $G^i$ contains at least one simple $v_i$-cycle of length $k$ or less.
For each graph $G^i$, this can be decided by a distance-limited breadth-first search
from $v_i$, testing if any predecessor of $v_i$ can be reached within distance $k-1$.
This takes $O(n+m)$ time per graph and the total pre-processing time is $O(n(n+m))$.
After pre-processing, for all $v_i \in L$, $\bsdfs(G^i, v_i, v_i, k)$ is executed,
producing all simple cycles in $G$ of length $\le k$.

Further pre-processing is common for reducing the practical
runtime~\cite{DBLP:journals/siamcomp/Johnson75,
DBLP:journals/corr/abs-2105-10094, peng_efficient_2021}: reducing $G^i$ to
the strongly connected component containing $v_i$, or removing nodes $x$ with
$dist^{G^i}(v_i, x) + dist^{G^i}(x, v_i) > k$,
which cannot lie on any qualifying cycle.

\medskip
In the remainder of the paper we focus on the path-search setting and,
unless stated otherwise, assume that $s \neq t$.
The assumption is not restrictive: the case $s = t$ is covered verbatim
by \fullref{lem:cycle-reduction}, and it lets us state path-avoidance
conditions without endpoint exemptions.

\subsection{Basic Observations}

The following elementary observations capture structural properties and
invariants of $\bsdfs$ that will be used repeatedly in the
correctness proofs. Their proofs are straightforward and are therefore
omitted.

\begin{enumerate}[label=(\arabic*),ref=\arabic*]

	\item\label{obs:simple-search-path}
	Every search path is simple.

	\item\label{obs:simple-output}
	Every output path is a simple $s$-$t$-path.

	\item\label{obs:search-path-uniqueness}
	Every search path occurs at most once.

	\item\label{obs:no-duplicate-output}
	No output path is produced more than once.

	\item\label{obs:return-value}
	Every return value $sd$ is non-negative.

	\item\label{obs:path-barrier-fixed}
	During a call $\search(v)$, the barrier value $b[v]$ can be modified only by that call's final assignment.

	\item\label{obs:source-barrier}
	For the source, $b[s]=0$ remains during execution until the final assignment of the initial call.

	\item\label{obs:target-barrier}
	For the target, $b[t]=0$ remains throughout the execution,
	except in the case $s=t$ where \autoref{obs:source-barrier} applies.

	\item\label{obs:cascade-lowers}
	In $\fruitful$, every cascade assignment strictly decreases the barrier value of the updated node.

\end{enumerate}

\section{Soundness}\label{sec:soundness}

In this section we show that every output produced by $\bsdfs(G, s, t, k)$ is a
simple $s$-$t$-path of length $\le k$ in $G$.

\begin{lemma}[Parent Pruning Guard]\label{lem:parent-pruning-guard}
	During $\bsdfs(G, s, t, k)$, every call $\search(v)$ with search path $Pv$ satisfies at entry
	$$b[v] \le k - h, \qquad h = \|Pv\|.$$
\end{lemma}
\begin{proof}
	By induction over the search-path length $h$.
	\emph{Base case} ($h = 0$): $v = s$. By Observation~\ref{obs:source-barrier} $b[s] = 0 \le k$ at entry.

	\emph{Inductive step}: $\search(v)$ with search path $Pv$ is invoked from
	some parent call $\search(u)$ with search path $P = P' \cdot u$ of length $\|P\| = h - 1$.
	The recursion is reached only when the parent's pruning condition holds for $v$: $b[v] + (h-1) < k$, i.e.\ $b[v] < k - h + 1$, hence $b[v] \le k - h$.
\end{proof}

\begin{lemma}[Barrier non-negative]\label{lem:barrier-ge-0}
	During $\bsdfs(G, s, t, k)$ holds
	\begin{enumerate}
		\item $b[x] \ge 0$ for all $x \in V$,
		\item after initialization, only strictly positive values are assigned.
	\end{enumerate}
\end{lemma}
\begin{proof}
	We argue by induction over the sequence of assignments to entries of $b$;
	recall that no assignment changes $b[v]$ while $v$ is on the search path (Observation~\ref{obs:path-barrier-fixed}).
	\begin{enumerate}
		\item Initialization: $b[x] \gets 0$ for all $x \in V$.
		\item Fruitful assignment $b[v] \gets sd$: the value $sd$ is either $1$, set when an
		edge to $t$ is found, or $d + 1$ for some child return value $d \ge 0$
		(Observation~\ref{obs:return-value}); in either case $sd \ge 1$.
		\item Fruitless assignment $b[v] \gets k - h + 1$: by the inductive hypothesis and
		\fullref{lem:parent-pruning-guard}, at entry $0 \le b_{entry}[v] \le k - h$, so
		the assigned value $k - h + 1 > k - h \ge 0$.
		\item Cascade assignment $b[p] \gets d + 1$: the cascade starts from $(v, sd)$ with
		$sd \ge 1$ (case 2), and each enqueued pair carries a strictly larger value
		than its dequeued $(q, d)$; so every assigned $d + 1 \ge 2 > 0$.
	\end{enumerate}
	In any case but initialization, a strictly positive value is assigned.
\end{proof}

\begin{lemma}[Search Path Length]\label{lem:search-path-length}
	In $\bsdfs(G, s, t, k)$, every call $\search(v)$ has search-path length $h = \|Pv\|$ with $0 \le h \le k$.
\end{lemma}
\begin{proof}
	By \fullref{lem:parent-pruning-guard}, $b[v] \le k - h$, so $h + b[v] \le k$; and by \fullref{lem:barrier-ge-0}, $h \le h + b[v] \le k$.
\end{proof}

\begin{lemma}[Termination]\label{lem:termination}
	The algorithm $\bsdfs(G,s,t,k)$ always terminates.
\end{lemma}
\begin{proof}
	We consider all places where loops or recursion occur.

	\emph{For-loops.}
	In each call $\search(v)$ the for-loop iterates once over the finite set $\suc(v)$, and no successor is reconsidered;
	likewise each step of the $\fruitful$ cascade iterates once over $\pre(q)$ for the dequeued node $q$.
    Thus, every for-loop is finite.

	\emph{Search recursion.}
	By Observation~\ref{obs:simple-search-path} the search path is always simple, and by Observation~\ref{obs:search-path-uniqueness}
    each simple path occurs at most once.
	Since a finite graph has finitely many simple paths, only finitely many recursive calls $\search(\cdot)$ occur.

	\emph{Cascade.}
	Each assignment $b[p] \gets d + 1$ in $\fruitful$ strictly decreases the potential $\sum_{x \in V} b[x]$ (the guard ensured $b[p] > d + 1$),
	which by \fullref{lem:barrier-ge-0} is bounded below by $0$; so only finitely many assignments occur.
	As a pair is enqueued only together with such an assignment, the queue is emptied after finitely many steps.

	Since the cascades, the search recursion, and all for-loops are finite, the
	entire execution of $\bsdfs(G,s,t,k)$ is finite.
\end{proof}

Summarizing all above claims, we conclude that the algorithm produces only
valid output (soundness):

\begin{theorem}[$\bsdfs$ Soundness]\label{thm:bsdfs-soundness}
	Every path produced by $\bsdfs(G, s, t, k)$ is a simple $s$-$t$-path in $G$ of length $\le k$.
\end{theorem}
\begin{proof}
	Output is produced in a call $\search(v)$ when a successor $w = t$ passes the pruning condition;
	the output is $Pv \cdot t$, which by Observation~\ref{obs:simple-output} is a simple $s$-$t$-path in $G$.
	Its length is $\|Pv\| + 1 = h + 1$; the pruning condition $b[t] + h < k$ held,
    and $b[t] = 0$ by Observation~\ref{obs:target-barrier}, so $h < k$ and the length is at most $k$.
	By \fullref{lem:termination}, $\bsdfs(G, s, t, k)$ is finite.
\end{proof}

We close this section by stating two corollaries.

\begin{corollary}[Fruitful Return]\label{cor:fruitful-return}
	If a call $\search(v)$ returns a \emph{fruitful value} $sd \le k$, then in fact $sd \le k - h$, where $h = \|Pv\|$.
\end{corollary}
\begin{proof}
	By induction on the search recursion, well-founded by \fullref{lem:termination}.

	\emph{Base case:}
	a successor $w = t$ passes the pruning condition $b[t] + h < k$ in $\search(v)$, so $sd \gets 1$ is executed and the returned value satisfies $sd \le 1$.
	Since $b[t] = 0$ by Observation~\ref{obs:target-barrier}, the pruning condition gives $h < k$, hence $sd \le 1 \le k - h$.

	\emph{Inductive step:}
	no successor $w = t$ passes the pruning condition; then the fruitful value $sd$ is set by a recursive call.
	At least one child returned a fruitful value $d$; let $w$ be a child achieving the minimum such $d$.
	Its search path has length $h + 1$, so by the inductive hypothesis $d \le k - (h+1)$.
	In $\search(v)$ we set $sd = d + 1$, hence $sd \le k - h$.
\end{proof}

\begin{corollary}[Barrier Upper Bound]\label{cor:barrier-upper}
    Throughout the execution of $\bsdfs(G,s,t,k)$,
    \[
        0 \le b[x] \le k \quad \text{for all } x \in V,
    \]
    except that a completely fruitless run sets $b[s] = k+1$ as its last action.
    Thus, $k+1$ is a sentinel value, strictly above every barrier value that is ever read.
\end{corollary}
\begin{proof}
    The lower bound is \fullref{lem:barrier-ge-0}.
    For the upper bound we induct over the sequence of assignments to $b$, showing each assigned value is $\le k$ apart from the stated exception.
    \begin{enumerate}
        \item Initialization: $b[x] \gets 0 \le k$.
        \item Fruitful assignment: $b[v] \gets sd \le k - h \le k$ by
              \fullref{cor:fruitful-return}.
        \item Fruitless assignment $b[v] \gets k - h + 1$: for $h \ge 1$ this is $\le k$; for
              $h = 0$, i.e.\ $v = s$, it is $k+1$, which occurs only when the initial call
              $\search(s)$ is fruitless and is then the algorithm's last
              action.
        \item Cascade assignment $b[p] \gets d + 1$:
                The upper bound follows from the inductive hypothesis and Observation~\ref{obs:cascade-lowers}
    \end{enumerate}
    Hence, every value is $\le k$ except the terminal $b[s] = k+1$;
    as that assignment is the algorithm's last action, every barrier read during the execution sees a value $\le k$.
\end{proof}

\section{Completeness}\label{sec:completeness}

In this section we prove the completeness of $\bsdfs(G,s,t,k)$, namely that it produces every simple $s$-$t$-path in $G$ of length at most $k$.
The proof is based on a graph-local property of the barrier values, called \emph{edge-consistent labeling}.
We first introduce this property and establish its basic consequences.
We then show that it is maintained throughout the execution of $\bsdfs$, and finally use it to prove completeness.

\subsection{Edge-Consistent Labeling}

The following discussion is naturally interpreted through the framework of heuristic search.
The barrier value $b[x]$ may be viewed as an estimate of the remaining distance from $x$ to the target $t$,
while the search-path length $h$ represents the cost already incurred.
In the terminology of $A^*$-search~\cite{DBLP:journals/tssc/HartNR68}, $h$ plays the role of the $g$-value and $b[x]$ the role of a heuristic estimate.
The pruning condition $b[w] + h < k$ is then $f$-value pruning against a budget:
a successor is explored only when the estimated total cost $g + h$ stays below
$k$. The edge-consistent labeling is precisely the \emph{consistency}
(monotonicity) condition $h(x) \le w(x,y) + h(y)$ of $A^*$, with unit edge
weights $w \equiv 1$ and required only for edges whose two endpoints lie outside the current search path.

The completeness argument then follows the classical pattern.
Consistency implies admissibility: the barrier never overestimates the true remaining distance to the target.
We shall show that $b[x]\le dist_S(x,t)$, where $S$ is the current search path.
Consequently, whenever a simple $s$-$t$-path of length at most $k$ exists through a successor, the pruning condition cannot exclude that successor.
Every feasible path therefore survives all pruning tests and is eventually generated.

The principal difference from classical heuristic search is that the heuristic is not fixed. In $A^*$,
consistency is a static property verified once for a given heuristic.
In contrast, the barrier values of $\bsdfs$ are continually modified during the search:
fruitless returns raise barriers, while update cascades may lower them.
Consistency is therefore not a precondition but an invariant that must be preserved under every state update.
Establishing this invariant constitutes the main technical component of the completeness proof.
In this respect $\bsdfs$ is closer to \emph{learning
    real-time} $A^*$~\cite{DBLP:journals/ai/Korf90}, where the heuristic is
likewise refined as the search proceeds, than to the static algorithm.

\begin{definition}[Edge-Consistent labeling]\label{def:edge-consistent-labeling}

    Let $G=(V,E)$ be a directed graph and let $S\subseteq V$.
    A function $b:V\rightarrow\mathbb N_0$ is an
    \emph{edge-consistent labeling with respect to $S$} if
    \[
    b[x]\le b[y]+1
    \]
    for every edge $(x,y)\in E$ satisfying
    $\{x,y\}\cap S=\emptyset$.
\end{definition}

\begin{lemma}[Barrier Distance Bound]\label{lem:bar-invariant}
    Let $G=(V,E)$ be a directed graph, $S \subseteq V$,
    and $b: V \rightarrow \mathbb{N}_0$ an edge-consistent labeling w.r.t.\ $S$. Then
    $$b[x] \le b[y] + l \quad \text{for every $x$-$y$-path $P$ of length $l$ in $G$ with $P \cap S = \emptyset$;}$$
    in particular, $b[x] - b[y] \le dist_S(x,y)$ for all $x, y \in V \setminus S$.
\end{lemma}
\begin{proof}
    By induction on $l$. For $l = 0$, $x = y$ and the claim holds.
    In the inductive step, let $(z,y) \in E$ be the last edge of the path.
    The induction hypothesis applied to the prefix up to $z$ gives $b[x] \le b[z] + l - 1$,
    and \fullref{def:edge-consistent-labeling} gives $b[z] \le b[y] + 1$; substitution yields $b[x] \le b[y] + l$.
    The distance bound follows by taking a shortest $x$-$y$-path avoiding $S$
    (trivially when $dist_S(x,y) = \infty$).
\end{proof}

Specialized to $y = t$ with $b[t] = 0$, the distance bound reads $b[x] \le dist_S(x,t)$:
the barrier never overestimates the true remaining distance --- \emph{admissibility}
in the sense of $A^*$. What makes the analysis nontrivial is that $S$ changes
during the search, and admissibility is a property \emph{of a given $S$}; the
remainder of this section tracks how the three kinds of state change --- barrier
writes, path pushes, and path pops --- affect it.

We stress that edge-consistency is sufficient but stronger than necessary:
completeness needs only that the barrier be admissible ($b[x] \le dist_S(x,t)$) at the
moment each successor is tested, which is all \fullref{lem:pruning-is-permissive} uses.
Any labeling that under-estimates the remaining distance at that instant therefore
yields completeness, edge-consistent or not; we maintain the stronger, local
edge-consistency only because it \emph{implies} admissibility incrementally across the
search's writes, pushes, and pops.

\begin{lemma}[Pruning is Permissive]\label{lem:pruning-is-permissive}
    Let $b$ be edge-consistent with respect to a set $S$, with $b[t] = 0$, and let
    $P = (v_0, \dots, v_h, v_{h+1}, \dots, v_l = t)$ be a simple path of length $l \le k$
    ending at $t$ whose prefix has node set $S = \{v_0, \dots, v_h\}$.
    Then $b[v_{h+1}] + h < k$.
\end{lemma}
\begin{proof}
    Since $P$ is simple, the suffix $(v_{h+1}, \dots, v_l = t)$ avoids
    $S = \{v_0, \dots, v_h\}$ and witnesses $dist_S(v_{h+1}, t) \le l - h - 1$.
    By \fullref{lem:bar-invariant} and $b[t] = 0$,
    $b[v_{h+1}] \le l - h - 1$, hence $b[v_{h+1}] + h \le l - 1 < l \le k$.
\end{proof}

\subsection{Fruitless and Fruitful Bounds}

The following lemmas characterize how a single search call affects the barrier values.
A fruitless call can only \emph{increase} its own barrier (\autoref{lem:fruitless-increasing}),
and while a call remains fruitless no barrier value decreases anywhere in its recursion subtree (\autoref{lem:fruitless-monotonicity}).
Conversely, a fruitful call returns a value that is at least the true distance to the target (\autoref{lem:fruitful-lower-bound}),
and the ensuing update cascade assigns barrier values that are likewise lower bounded
by the corresponding distances to $t$ (\autoref{lem:update-distance-bound}).
The complementary upper bounds, which together imply that cascade updates
coincide with the exact distances, are established in \autoref{sec:fruitful-cascades}.

Recall from \autoref{sec:algorithm} that during a call $\search(v)$ (\autoref{alg:search}) the search path is the current stack snapshot
$Pv=(s=v_0,\ldots,v_h=v)$
where $h = \|Pv\|$, and that $P$ denotes the prefix of this path at the time the call is entered.
Since Observation~\ref{obs:target-barrier} shows that $b[t]=0$ throughout the execution,
every lemma below may assume this property without further mention.

\begin{lemma}[Fruitless Increasing]\label{lem:fruitless-increasing}
    Suppose in $\bsdfs(G, s, t, k)$ a call $\search(v)$ returns fruitless.
    Then, for the barrier value of $v$ at entry $b_{entry}[v]$ and at exit $b_{exit}[v]$,
    $b_{exit}[v] > b_{entry}[v]$.
\end{lemma}
\begin{proof}
    By \fullref{lem:parent-pruning-guard}, $b_{entry}[v] \le k - h$ at entry, where $h = \|Pv\|$.
    A fruitless return executes the assignment $b[v] \gets k - h + 1$, which by Observation~\ref{obs:path-barrier-fixed}
    is the only write to $b[v]$ during $\search(v)$; hence
    $b_{exit}[v] = k - h + 1 > k - h \ge b_{entry}[v]$.

\end{proof}

\begin{lemma}[Fruitless Monotonicity]\label{lem:fruitless-monotonicity}
	If $\search(v)$ returns fruitless, then no call nested within $\search(v)$ returns fruitful,
    no execution of $\fruitful$ occurs during $\search(v)$, and every barrier value is non-decreasing throughout the execution of $\search(v)$.
\end{lemma}
\begin{proof}
	Suppose that some call nested within $\search(v)$ returns a fruitful value.
	By Corollary~\ref{cor:fruitful-return}, a fruitful call at depth $h'$ returns a value at most $k-h'$. Its parent therefore computes

	\[
	sd \le (k-h')+1 = k-(h'-1)\le k,
	\]
	and is itself fruitful. Repeating the argument up the recursion stack shows that $\search(v)$ must also be fruitful, contradicting the assumption.
	Therefore, no nested search call is fruitful. Since the procedure $\fruitful$ is invoked only after a fruitful return, no cascade execution can occur within $\search(v)$.
	The only barrier updates performed during $\search(v)$ are therefore fruitless assignments,
    either by $\search(v)$ itself or by calls nested within it.
    By \fullref{lem:fruitless-increasing}, each such update strictly increases the corresponding barrier value.
    Consequently, every barrier value is non-decreasing throughout the execution of $\search(v)$.
\end{proof}

\begin{lemma}[Fruitful Lower Bound]\label{lem:fruitful-lower-bound}
    If a call $\search(v)$ with search path $Pv$ returns a fruitful value $sd$, then

    \[
    sd \ge dist_{Pv}(v,t).
    \]
\end{lemma}

    \begin{proof}
    We proceed by induction on the search recursion, which is well-founded by \fullref{lem:termination}.
    A fruitful value $sd$ is obtained in one of two ways.
    If an edge $(v,t)$ is discovered, then $sd=1$. Since such an edge yields a path of length one,
    $dist_{Pv}(v,t)\le 1 = sd$.
    Otherwise,
    $sd=d_w+1$
    for some successor $w\notin Pv$ whose recursive call $\search(w)$ returns
    the minimum fruitful value $d_w$ encountered by the loop.
    By the induction hypothesis,
   $ d_w \ge dist_{Pvw}(w,t)$.
    Since $w$ is exempt as an endpoint,
    $dist_{Pvw}(w,t)=dist_{Pv}(w,t)$.
    Furthermore,
    $dist_{Pv}(v,w)\le 1$.
    Applying \fullref{lem:dist-triangle} to the intermediate node $w\notin Pv$ gives
    \[
    dist_{Pv}(v,t)
    \le dist_{Pv}(v,w)+dist_{Pv}(w,t)
    \le 1+d_w
    = sd.
    \]
\end{proof}

\begin{lemma}[Fruitful Distance Lower Bound]\label{lem:update-distance-bound}
    Let the cascade $\fruitful(v, sd)$ be invoked from a fruitful $\search(v)$ with search path $Pv$, prefix $P$.
    Then every assignment $b[u] \gets d+1$ performed by the cascade satisfies $b[u] \ge dist_P(u,t)$.
\end{lemma}
\begin{proof}
    The cascade does not modify the search path; $P$ and $Pv$ are fixed throughout.
    We show by induction on enqueue order that every pair $(q,d)$ placed on the
    queue satisfies $d \ge dist_P(q,t)$. Since each assignment $b[u] \gets d+1$
    is accompanied by enqueuing $(u, d+1)$, the claim follows.

    The initial pair is $(v, sd)$. By \fullref{lem:fruitful-lower-bound},
    $sd \ge dist_{Pv}(v,t)$, and since $v$ is endpoint-exempt,
    $dist_{Pv}(v,t) = dist_P(v,t)$.

    Every later pair $(u, d+1)$ is enqueued while an earlier pair $(q,d)$ with
    $u \in \pre(q)$ is processed; by the induction hypothesis,
    $d \ge dist_P(q,t)$. The intermediate node $q$ lies off $P$: either
    $q = v \notin P$, or $q$ passed the guard $q \notin S$ when enqueued, and
    $S = Pv \supseteq P$ throughout the cascade.
    Hence, \fullref{lem:dist-triangle} on the edge $(u,q)$ gives
    $dist_P(u,t) \le 1 + dist_P(q,t) \le d + 1 = b[u]$.
\end{proof}

\subsection{Fruitful Cascades}\label{sec:fruitful-cascades}

\begin{lemma}[Cascade Distance]\label{lem:cascade-distance}
  	Suppose the cascade $\fruitful(v,sd)$ is initiated by a $\search(v)$ call whose search path is $Pv$,
    and assume that the barrier labeling is edge-consistent with respect to $Pv$ at cascade entry.
    Then, at cascade exit,
	\[
	b[x] \le sd + dist_{Pv}(x,v)
	\quad
	\text{for every } x \notin P,
	\]
	with equality for every node that was enqueued by the cascade.
\end{lemma}
\begin{proof}
    By Observation~\ref{obs:cascade-lowers}, the cascade only decreases barrier values and,
    because nodes are processed in FIFO order, each node is relaxed at most once.

    \emph{Upper bound.}
    We prove $b[x] \le sd + dist_{Pv}(x,v)$ for every $x \notin P$ by induction on $r = dist_{Pv}(x,v)$.
    The claim is trivial when $r=\infty$.
    For $r=0$, we have $x=v$, and the cascade initializes $b[v]=sd$.

For $r \ge 1$ we have $x \ne v$, and $x \notin P$ by assumption, so
    $x \notin Pv = S$ throughout the cascade. Let $x_1$ be the successor of $x$ on a
    shortest $x$-$v$-path avoiding $Pv$ except at $v$, so $dist_{Pv}(x_1,v) = r-1$ and
    either $x_1 = v$ or $x_1 \notin Pv$; by induction $b[x_1] \le sd + (r-1)$ at exit.
    If $x_1$ is enqueued (in particular if $x_1 = v$), then its pair $(x_1, b[x_1])$ is
    dequeued and the predecessor $x$ is examined: it passes the on-path guard
    $x \notin S$, and whether or not the value guard $b[x] > b[x_1] + 1$ fires,
    $b[x] \le b[x_1] + 1 \le sd + r$ holds afterwards.
    If $x_1$ is never enqueued, $b[x_1]$ keeps its entry value and entry edge-consistency on the
    edge $(x, x_1)$ gives $b_{entry}[x] \le b[x_1] + 1 \le sd + r$; by Observation~\ref{obs:cascade-lowers},
    $b[x] \le b_{entry}[x] \le sd + r$ at exit.

    \emph{Lower bound for enqueued nodes.}
    If $u$ is enqueued, it is relaxed by some dequeued $(q, d)$ with $u \in \pre(q)$, setting $b[u] = d + 1$ where $d = b[q]$.
    By induction on dequeue order, $b[q] \ge sd + dist_{Pv}(q,v)$ (base $b[v] = sd$);
    either $q = v$, and the edge $(u,v)$ gives $dist_{Pv}(u,v) \le 1 = 1 + dist_{Pv}(v,v)$;
    or $q \notin Pv$, and the edge $(u,q)$ with \fullref{lem:dist-triangle} gives $dist_{Pv}(u,v) \le 1 + dist_{Pv}(q,v)$.
    Thus, $b[u] = b[q] + 1 \ge sd + dist_{Pv}(u,v)$. With the upper bound, $b[u] = sd + dist_{Pv}(u,v)$.
\end{proof}

We now isolate the only state transition that can violate admissibility.
Barrier \emph{writes} cannot do so: a fruitless or fruitful return modifies the barrier
of a node that still lies on the current search path,
and admissibility places no restriction on such nodes.
Likewise, a cascade only decreases barrier values (Observation~\ref{obs:cascade-lowers}),
so it cannot raise an admissible value above its distance bound.

A \emph{push}
enlarges the forbidden set, which only relaxes the edge-consistency constraints.
The remaining transition is a \emph{pop}: when $\search(v)$ removes $v$ from the
search path, the forbidden set shrinks from $Pv$ to $P$, and an edge $(x,y)$ is
newly constrained w.r.t.\ $P$ exactly when $v \in \{x,y\}$ --- removing $v$
changes the constrained/exempt status of no other edge. Re-establishing
edge-consistency at a pop is therefore a \emph{local} obligation on the edges
incident to $v$. The next lemma discharges the incoming edges and the case
$y = v$; the single family it defers --- edges leaving $v$ --- is handled in
\fullref{lem:search-preserves-consistency}, and is the only point where the
exact fruitful value is needed.

\begin{lemma}[Update Repairs]\label{lem:update-repairs-consistency}
    Under the cascade $\fruitful(v,sd)$ of  \fullref{lem:cascade-distance},
    edge-consistency with respect to $P$ holds at cascade exit for every edge except, possibly, those leaving $v$.
\end{lemma}
\begin{proof}
    Take an edge $(x, y)$ with $\{x,y\} \cap P = \emptyset$ and $x \ne v$;
    then $x \notin Pv$, and either $y = v$ or $y \notin Pv$. By \fullref{lem:cascade-distance},
    $b[x] \le sd + dist_{Pv}(x,v)$.

    \emph{$y = v$:} the edge $(x,v)$ gives $dist_{Pv}(x,v) \le 1$, so $b[x] \le sd + 1 = b[v] + 1$.

    \emph{$y \notin Pv$ and $y$ is enqueued:} then $b[y] = sd + dist_{Pv}(y,v)$ by \fullref{lem:cascade-distance},
    and \fullref{lem:dist-triangle} ($y \notin Pv$) with $dist_{Pv}(x,y) \le 1$ gives
    \[
        b[x] \;\le\; sd + dist_{Pv}(x,v) \;\le\; sd + 1 + dist_{Pv}(y,v) \;=\; b[y] + 1.
    \]

    \emph{Case $y\notin Pv$ and $y$ is not enqueued.}
    Since $y$ is never enqueued, the cascade never writes to $b[y]$.
    Hence
    $b[y] = b_{\mathrm{entry}}[y].$ Entry edge-consistency on $(x,y)$ gives
    \[
    b_{\mathrm{entry}}[x]
    \le
    b_{\mathrm{entry}}[y]+1.
    \]
    Since the cascade only lowers barriers,
    \[
    b[x]
    \le
    b_{\mathrm{entry}}[x]
    \le
    b_{\mathrm{entry}}[y]+1
    =
    b[y]+1.
    \]

\end{proof}

\subsection{Edge-Consistency of $\bsdfs$}
It suffices to show that $\bsdfs$ keeps its labeling edge-consistent w.r.t.\ the path after the pop;
completeness then follows. By the preceding subsection the only nontrivial
obligation at a pop of $v$ is on the edges \emph{leaving} $v$, and discharging it
needs the cascade to have deposited the \emph{exact} distance --- equivalently,
that the fruitful return value satisfies $sd = dist_{Pv}(v,t)$, which in turn
needs the shortest completion to have been produced. Completeness, the exact
value, and preservation at the pop are thus mutually dependent and are
established \emph{together}, by a single induction on the search recursion,
well-founded by \fullref{lem:termination}.

For each call $\search(v)$ the induction carries the following clauses:
\begin{enumerate}
\item[(i)]      completeness of the call (\autoref{lem:call-complete});
\item[(ii)]     the exact fruitful value (\autoref{lem:strict-bar-invariant});
\item[(iii)]    global monotonicity (\autoref{lem:per-call-monotone});
\item[(iv)]     preservation across the call (\autoref{lem:search-preserves-consistency}).
\end{enumerate}

In each of the four lemmas below the \emph{induction hypothesis} is that all four claims
hold for every call nested strictly within the current $\search(v)$.
A pruning test that a budget-respecting path must survive is discharged directly by
\fullref{lem:pruning-is-permissive}; no global completeness statement is invoked inside the induction.

\begin{lemma}[Call Completeness]\label{lem:call-complete}
    Consider a call $\search(v)$ with search path $Pv = (s = v_0, \dots, v_h = v)$
    at which edge-consistency w.r.t.\ $P$ holds at entry.
    Then $\search(v)$ produces every simple $s$-$t$-path of length $\le k$ having $Pv$ as a prefix.
\end{lemma}
\begin{proof}
    We argue within the joint induction on the search recursion; the induction hypothesis
    supplies all four claims for every call nested strictly within $\search(v)$.

    On entry $v$ is pushed, turning $P$ into $Pv$; pushing only
    enlarges the set of exempt nodes, so edge-consistency w.r.t.\ $Pv$ holds. Within
    the for-loop, a successor that is pruned or equals $t$ writes no barrier, while
    a recursively searched successor $w'$ is entered with search path $Pvw'$ and,
    by the induction hypothesis
    (\autoref{lem:search-preserves-consistency}), preserves edge-consistency
    w.r.t.\ $Pvw' \setminus \{w'\} = Pv$ across the call.
    Edge-consistency w.r.t.\ $Pv$ is therefore maintained throughout the for-loop,
    and in particular holds whenever a successor is examined.

    Let $R = (s = v_0, \dots, v_h = v, v_{h+1}, \dots, v_l = t)$ be a simple
    $s$-$t$-path of length $l \le k$ with prefix $Pv$, and let $w = v_{h+1}$. As $w \in \suc(v)$,
    the for-loop examines it; at that moment edge-consistency w.r.t.\ $Pv$ holds
    and $b[t] = 0$ (Observation~\ref{obs:target-barrier}), so by
    \fullref{lem:pruning-is-permissive} $b[w] + h < k$ and $w$ is not pruned.
    If $w = t$, the closing-edge branch outputs $Pv$ followed by $t$, which is $R$.
    Otherwise, $w \ne t$ and, $R$ being simple, $w \notin \{v_0, \dots, v_h\} = Pv$;
    hence $\search(w)$ is entered with search path $Pvw$, at which
    edge-consistency w.r.t.\ $Pv$ holds.
    Since $R$ has prefix $Pvw$, the induction hypothesis
    (\autoref{lem:call-complete} for $\search(w)$) yields that $R$ is produced.
\end{proof}

\begin{lemma}[Strict Barrier Invariant]\label{lem:strict-bar-invariant}
    Consider a call $\search(v)$ with search path $Pv = (s = v_0, \dots, v_h = v)$
    at which edge-consistency w.r.t.\ $P$ holds at entry.
    If $\search(v)$ is fruitful, then $sd = dist_{Pv}(v,t)$.
\end{lemma}
\begin{proof}
    Again within the joint induction, with the induction hypothesis available for every nested call.
    The bound $sd \ge dist_{Pv}(v,t)$ is \fullref{lem:fruitful-lower-bound}; it remains to show $sd \le dist_{Pv}(v,t)$.

    Being fruitful, $\search(v)$ has $sd \le k - h$
    (\fullref{cor:fruitful-return}), so $dist_{Pv}(v,t) \le sd \le k - h$ is finite.
    Let $R = (v, w, \dots, t)$ be a shortest $v$-$t$-path with $R \cap Pv \subseteq \{v,t\}$,
    of length $dist_{Pv}(v,t)$, and let $w$ be its node after $v$.
    Prepending $Pv$ yields a simple $s$-$t$-path of length $h + dist_{Pv}(v,t) \le k$
    with prefix $Pv$ whose node after $v$ is $w$.

    After $v$ is pushed, edge-consistency w.r.t.\ $Pv$ holds, and by the induction
    hypothesis (\autoref{lem:search-preserves-consistency}) each successor searched
    before $w$ preserves it; so it holds when $\search(v)$ examines $w$.
    With $b[t] = 0$ (Observation~\ref{obs:target-barrier}), \fullref{lem:pruning-is-permissive}
    gives $b[w] + h < k$, so $w$ is not pruned.

    If $w = t$, then $dist_{Pv}(v,t) = 1$ and the closing-edge branch sets $sd \gets 1 = dist_{Pv}(v,t)$.
    Otherwise, $w \ne t$ and, $R$ being simple with $R \cap Pv \subseteq \{v,t\}$, $w \notin Pv$;
    so $\search(w)$ runs with search path $Pvw$.
    The $s$-$t$-path above has prefix $Pvw$, so by the induction hypothesis (\autoref{lem:call-complete}
    for $\search(w)$) it is produced during $\search(w)$,
    which is therefore fruitful (\fullref{lem:fruitless-monotonicity});
    by the induction hypothesis
    (\autoref{lem:strict-bar-invariant} for $\search(w)$) its return value is $sd_w = dist_{Pvw}(w,t)$.
    The suffix of $R$ from $w$ meets $Pvw$ only at $w$ and $t$ and has length
    $dist_{Pv}(v,t) - 1$, so $dist_{Pvw}(w,t) \le dist_{Pv}(v,t) - 1$.
    The for-loop sets $sd \le sd_w + 1 = dist_{Pvw}(w,t) + 1 \le dist_{Pv}(v,t)$.

    In both cases $sd \le dist_{Pv}(v,t)$, hence $sd = dist_{Pv}(v,t)$.
\end{proof}

\begin{lemma}[Per-Call Global Monotonicity]\label{lem:per-call-monotone}
    Consider a call $\search(v)$ with search path $Pv = (s = v_0, \dots, v_h = v)$
    at which edge-consistency w.r.t.\ $P$ holds at entry, with entry barriers
    $b_{entry}$ and exit barriers $b_{exit}$. Then $b_{exit}[x] \ge b_{entry}[x]$
    for every $x \in V$.
\end{lemma}
\begin{proof}
    Within the joint induction; the induction hypothesis holds for every nested call.
    The execution of $\search(v)$ writes barriers only through (i) its child calls,
    (ii) the post-loop assignment at $v$, and, in the fruitful case, (iii) the
    cascade $\fruitful(v, sd)$.

    \emph{For-loop.}
    On entry $v$ is pushed and edge-consistency w.r.t.\ $Pv$ holds; a recursively
    searched successor $w'$ is entered with search path $Pvw'$ and, by the
    induction hypothesis (\autoref{lem:search-preserves-consistency}), preserves
    edge-consistency w.r.t.\ $Pv$ across the call. Each child is thus entered with
    the invariant and is, by the induction hypothesis
    (\autoref{lem:per-call-monotone}), globally non-decreasing; chaining over the
    children in order, $b[x] \ge b_{entry}[x]$ for every $x$ when the for-loop completes.

    \emph{Fruitless case.}
    No cascade runs and the only post-loop write is $b[v] \gets k - h + 1$, which
    by \fullref{lem:parent-pruning-guard} exceeds $b_{entry}[v] \le k - h$.

    \emph{Fruitful case.}
    By \fullref{lem:bar-invariant} at entry ($v, t \notin Pv \setminus \{v\}$) with $b[t] = 0$
    (Observation~\ref{obs:target-barrier}) $b_{entry}[v] \le dist_P(v,t) = dist_{Pv}(v,t)$
    (dist is unaffected by whether its first argument lies in the forbidden set),
    and by \fullref{lem:fruitful-lower-bound} the post-loop write is
    $b[v] \gets sd \ge dist_{Pv}(v,t)$; so $b[v]$ does not decrease.
    The cascade then writes only nodes $u \notin Pv$, each assignment being
    $b[u] \gets d + 1 \ge dist_P(u,t)$ by \fullref{lem:update-distance-bound},
    while entry edge-consistency and \fullref{lem:bar-invariant} give
    $b_{entry}[u] \le dist_P(u,t)$.
    Hence, no node barrier falls below its entry value.
\end{proof}

\begin{lemma}[Search Preserves]\label{lem:search-preserves-consistency}
    Consider a call $\search(v)$ with search path $Pv = (s = v_0, \dots, v_h = v)$
    at which edge-consistency w.r.t.\ $P$ holds at entry.
    Then, edge-consistency w.r.t.\ $P$ holds at exit.
\end{lemma}
\begin{proof}
    Within the joint induction; the induction hypothesis holds for every nested call,
    and \fullref{lem:strict-bar-invariant} and \fullref{lem:per-call-monotone}
    are available for $\search(v)$ itself.

    \emph{Push and for-loop.}
    On entry $v$ is pushed, turning $P$ into $Pv$; pushing only enlarges the set of
    exempt nodes, so edge-consistency w.r.t.\ $Pv$ holds. A pruned successor and the closing-edge
    branch write no barrier, while a recursively searched successor $w'$ is entered with search path
    $Pvw'$ and, by the induction hypothesis, preserves edge-consistency w.r.t. $Pvw' \setminus \{w'\} = Pv$.
    Hence, edge-consistency w.r.t.\ $Pv$ holds when the for-loop completes.

    An edge $(x,y)$ is constrained by edge-consistency w.r.t.\ $P$
    exactly when $\{x,y\} \cap P = \emptyset$; such edges are of
    three kinds:
    \begin{enumerate}
        \item[(I)] edges with $v \notin \{x,y\}$;
        \item[(II)] incoming edges $(x,v)$ with $x \notin Pv$;
        \item[(III)] outgoing edges $(v,y)$ with $y \notin Pv$.
    \end{enumerate}
    Popping $v$ writes no barrier, so $b_{exit}$ denotes the values after the post-loop assignment
    and, in the fruitful case, after the cascade.

    \textbf{Fruitful case ($sd \le k$).}
    The assignment $b[v] \gets sd$ is made, $\fruitful(v, sd)$ runs, and $b_{exit}[v] = sd$ (the
    cascade does not write $b[v]$, as $v \in Pv$ throughout it, Observation~\ref{obs:path-barrier-fixed}).

    \emph{(I) and (II).}
    When the for-loop completes, edge-consistency w.r.t.\ $Pv$ holds; $b[v] \gets sd$ changes only
    $b[v]$, affecting only edges incident to $v$, which are exempt w.r.t.\ $Pv$, so edge-consistency
    w.r.t.\ $Pv$ still holds when $\fruitful(v,sd)$ is invoked. By
    \fullref{lem:update-repairs-consistency}, at the cascade's end edge-consistency w.r.t.\ $P$
    holds everywhere except possibly on edges leaving $v$;
    edges of kind (I) and (II) do not leave $v$, hence are consistent at exit.

    \emph{(III).}
    Let $(v,y)$ be an edge with $y \notin Pv$.
    If $y = t$: as $\search(v)$ is fruitful, $1 \le sd \le k - h$ (\fullref{cor:fruitful-return}),
    so $h < k$ and the closing edge $(v,t)$ passes its pruning test ($b[t] + h = h < k$), firing and
    setting $sd \gets 1$; with all return values $\ge 1$, $sd = 1$. With $b[t] = 0$
    (Observation~\ref{obs:target-barrier}), $b_{exit}[v] = 1 = b[t] + 1$.

    Now let $y \ne t$. By \fullref{lem:strict-bar-invariant} and
    endpoint-exemption, $sd = dist_{Pv}(v,t) = dist_P(v,t)$. The
    edge $(v,y)$ gives $dist_P(v,y) \le 1$, and
    \fullref{lem:dist-triangle} (intermediate $y \notin Pv \setminus \{v\}$) yields
    $sd \le 1 + dist_P(y,t)$, so $dist_P(y,t)
        \ge sd - 1$. We show $b[y] \ge sd - 1$ from the moment the for-loop finishes
    processing $y$ until exit.

    \emph{Base.}
    If $y$ was pruned, $b[y] + h < k$ failed, so $b[y] \ge k - h \ge sd$
    (\fullref{cor:fruitful-return}). If $y$ was recursively searched, then right after the child
    $\search(y)$ (depth $h+1$) returns, $b[y] = k - h \ge sd$ if that child was fruitless, or
    $b[y] = sd_y$ if it was fruitful, where $sd \le sd_y + 1$ (the for-loop minimization), so
    $b[y] = sd_y \ge sd - 1$. In all cases $b[y] \ge sd - 1$ when the for-loop finishes with $y$.

    \emph{Preservation.}
    After the for-loop finishes processing $y$, every further write to $b[y]$
    occurs either within a later child call of $\search(v)$ or in the post-loop
    cascade $\fruitful(v, sd)$; in particular, a cascade $\fruitful(z, sd_z)$ with
    $z$ a strict descendant of $v$ runs entirely within the child subtree
    containing $z$. Each later child is entered with edge-consistency w.r.t.\ $Pv$
    (as above) and is, by the induction hypothesis
    (\autoref{lem:per-call-monotone}), globally non-decreasing, so
    $b[y] \ge sd - 1$ persists through the end of the for-loop.
    The post-loop cascade writes $b[y] \gets d + 1 \ge dist_P(y,t) \ge sd - 1$ by \fullref{lem:update-distance-bound}.
    Hence, $b_{exit}[y] \ge sd - 1$, i.e.\ $b_{exit}[v] = sd \le b_{exit}[y] + 1$.

    \textbf{Fruitless case ($sd = k + 1$).}
    Then $b_{exit}[v] = k - h + 1$, and by \fullref{lem:fruitless-monotonicity} no cascade runs
    within $\search(v)$ and every barrier is non-decreasing throughout it.

    \emph{(I).}
    When the for-loop completes, edge-consistency w.r.t.\ $Pv$ holds; $b[v] \gets k - h + 1$ changes
    only $b[v]$, so every edge $(x,y)$ with $v \notin \{x,y\}$ and $\{x,y\} \cap Pv = \emptyset$ keeps
    $b[x] \le b[y] + 1$, and for such edges $\{x,y\} \cap Pv = \emptyset$ coincides with
    $\{x,y\} \cap P = \emptyset$.

    \emph{(II).}
    Let $(x,v)$ be an edge with $x \notin Pv$. At entry it is constrained w.r.t.\ $P$,
    so $b_{entry}[x] \le b_{entry}[v] + 1 \le (k - h) + 1$ (\fullref{lem:parent-pruning-guard}).
    During $\search(v)$, $b[x]$ changes only through a fruitless $\search(x)$ nested within it (depth
    $h_x \ge h + 1$, assigning $b[x] \gets k - h_x + 1 \le k - h$); so $b[x]$ never exceeds $k - h + 1$,
    giving $b_{exit}[x] \le k - h + 1 = b_{exit}[v] \le b_{exit}[v] + 1$.

    \emph{(III).}
    Let $(v,y)$ be an edge with $y \notin Pv$; we show $b_{exit}[y] \ge k - h$, which gives
    $b_{exit}[v] = k - h + 1 \le b_{exit}[y] + 1$. In the for-loop $y$ was pruned or searched.
    If pruned, $b[y] + h < k$ failed, so $b[y] \ge k - h$. If searched, the child $\search(y)$ has
    depth $h + 1$ and, being fruitless, assigns $b[y] \gets k - (h+1) + 1 = k - h$. Either way
    $b[y] \ge k - h$ at some point, and since barriers are non-decreasing, $b_{exit}[y] \ge k - h$.

    In all cases edge-consistency w.r.t.\ $P$ holds at exit.
\end{proof}

\subsection{Completeness of $\bsdfs$}

We have seen that edge-consistency w.r.t.\ the current search path holds throughout the execution of $\bsdfs(G,s,t,k)$.
This is the global reading of \fullref{lem:search-preserves-consistency}:
starting from $b \equiv 0$ (edge-consistent w.r.t.\ the empty path), the single call $\search(s)$ is
entered with the invariant, and the lemma shows each call preserves it w.r.t.\
the path after its node is popped; chaining over the call tree gives the invariant at every point.
We close with the main theorems in this section.

\begin{theorem}[$\bsdfs$ Completeness]\label{thm:bsdfs-complete}
    The algorithm $\bsdfs(G, s, t, k)$ produces every simple $s$-$t$-path of length $\le k$ in $G$.
\end{theorem}
\begin{proof}
    The initial call $\search(s)$ has search path $(s)$ and is entered with the labeling
    $b \equiv 0$, edge-consistent w.r.t.\ $(s) \setminus \{s\} = \emptyset$.
    Every simple $s$-$t$-path of length $\le k$ has $(s)$ as a prefix,
    so by \fullref{lem:call-complete} it is produced by $\search(s)$.
\end{proof}

Together with \fullref{thm:bsdfs-soundness} and
Observation~\ref{obs:no-duplicate-output}, this settles the enumeration
property: $\bsdfs(G,s,t,k)$ outputs simple $s$-$t$-paths of length $\le k$
only, outputs every one of them, and outputs none of them twice.

\section{Delay Bounds}\label{sec:delay-bounds}

For delay analysis, the entire $\bsdfs(G,s,t,k)$ execution timeline is partitioned
into an alternating sequence of \emph{phases}, delimited by the outputs and
by the returns of certain calls between two outputs. We first define
the phases and establish how individual $\search$ calls relate to them;
the counting then rests on a single fact --- within one phase, no node is
entered twice at the same depth --- from which the barrier-write bounds, the
call bounds, and both delay theorems follow. Coarser decompositions ---
\emph{intervals} between consecutive outputs and \emph{periods} between
consecutive phase boundaries --- are introduced when the accounting
requires them.

\subsection{Phases}\label{sec:phases}

Let the execution output $T$ paths and, for $\tau = 1, 2, \dots, T$, let
$S_\tau$ be the $\tau$-th path output and $o_\tau$ the execution step at
which $S_\tau$ is output. We call $S_\tau$ a \emph{spine} and write
$S_\tau = (s = v_0, v_1, \dots, v_l = t)$, suppressing the index $\tau$
on nodes and length when it is clear from context.

Consider the execution around $o_\tau$. The target test in
$\search(v_{l-1})$ recognizes the successor $t$ and the spine path is
output; no call frame for $t$ is created. Hence, the frames open at
$o_\tau$ are exactly $v_0, \dots, v_{l-1}$, at depths $0, \dots, l-1$.
$\search(v_{l-1})$ continues enumerating the successors of $v_{l-1}$ after
$t$ in its successor list and recurses into those that pass the pruning
condition.

It is convenient to bracket the execution:
Let $\beta_0 = o_0 :=$ the start of the execution and $o_{T+1} :=$ its
termination, the step immediately after the initial call $\search(s)$ has
returned, and set $S_0 := S_{T+1} := ()$, the empty spine. This is a
convention of the analysis, not a statement of the algorithm: neither $o_0$
nor $o_{T+1}$ produces a path, both cost no elementary step, and both lie in
the lifetime of no call --- in particular they render no call fruitful.
We call $o_0, o_1, \dots, o_{T+1}$ the \emph{events} of the execution;
the first is the start, the next $T$ are the outputs, the last is the termination.

Let $1 \le \tau \le T$ and let $S_{\tau+1}$ be the next spine. The two node
sequences have a longest common prefix $(v_0, \dots, v_{j})$, $j = j_\tau$;
we call $v_{j}$ the \emph{fork} of $S_\tau$ and $S_{\tau+1}$.
For $\tau < T$ both spines start at $s$, so $j_\tau \ge 0$; for $\tau = T$ the
next spine is empty, hence so is the common prefix, and we write
$j_T = -1$.
As long as the
current recursion level is $i > j$, all sibling searches emanating from
$\search(v_i)$ must, by definition of the fork, end fruitless. The final
step of each $\search(v_i)$ is its cascade $\fruitful(v_i, sd_i)$, since
$\search(v_i)$ itself is fruitful: it led to $S_\tau$.

Set $w_\tau := v_{j_\tau+1}$, the child of the fork on $S_\tau$, and
let $\beta_\tau$ be the moment after $o_\tau$ when $\search(w_\tau)$
returns; its cascade, run before returning, thus lies before $\beta_\tau$.
For $\tau < T$ it returns to its caller $\search(v_{j_\tau})$; for
$\tau = T$ we have $w_T = v_{j_T + 1} = v_0 = s$, so $\beta_T$ is the return
of the initial call --- the terminal unwinding read as the retreat towards
an empty next spine.
In the degenerate case $j_\tau = l - 1$ we have $w_\tau = t$, which has no
$\search$ call; we then set $\beta_\tau := o_\tau$. Finally, set
$\beta_{T+1} := o_{T+1}$.
By Observation~\ref{obs:search-path-uniqueness}, the frame of $w_\tau$
--- open at $o_\tau$ with search path $(v_0, \dots, v_{j_\tau+1})$ ---
occupies one contiguous span of the execution and returns exactly once,
so $\beta_\tau$ is well defined and $o_\tau \le \beta_\tau$ (equality only
in the degenerate case). Moreover, for $\tau < T$, $\search(v_{j_\tau})$
scans the successor branch containing $o_{\tau+1}$ only after the branch
through $w_\tau$ has returned, so $\beta_\tau < o_{\tau+1}$; and
$\beta_T < o_{T+1}$ by the definition of $o_{T+1}$.
The boundaries are thus ordered as
\[
    \beta_0 = o_0 < o_1 \le \beta_1 < o_2 \le \beta_2
    < \dots < o_T \le \beta_T < o_{T+1} = \beta_{T+1}.
\]
For $\tau = 1, \dots, T+1$, the
half-open segment $(\beta_{\tau-1}, o_\tau]$ is called \emph{ascent} $\tau$ and the
half-open segment $(o_\tau, \beta_\tau]$ \emph{descent} $\tau$; a descent may be
empty (the degenerate case above, and always for $\tau = T+1$), and so is
descent $0 := (o_0, \beta_0]$. The phases tile the execution;
each ascent ends with its event, and no descent contains an event.
The names resemble the way in which the spine $S_\tau$ transforms to spine $S_{\tau+1}$,
disregarding search path oscillations by intermediate fruitless sibling searches.

Edge cases. For $T > 0$ the terminal descent $T$ runs down to the return of
$\search(s)$, closing every spine frame of $S_T$; ascent $T+1$ then carries
the terminal event and nothing else. If $j_\tau = 0$, descent $\tau$
retreats to the root frame without closing it. If $T = 0$, the execution
consists of ascent $1$ alone, every call in it is fruitless, and its
descent is empty. Sums indexed by an empty spine are zero.

The next lemmas relate individual calls to the phases;
they carry the entire counting below.

\begin{lemma}[Descent Stack]\label{lem:descent-stack}
    Let $1 \le \tau \le T$.
    \begin{enumerate}
    \item At $\beta_\tau$ the open frames are exactly the
          frames $(v_0, \dots, v_{j_\tau})$ of the common prefix of
          $S_\tau$ and $S_{\tau+1}$, each pushed with its spine search
          path; for $\tau = T$ that prefix is empty, so no frame is open
          at $\beta_T$, and none is open at $\beta_0$ either.
    \item Throughout descent $\tau$, the search path has the form
          $(v_0, \dots, v_i, q_1, \dots, q_r)$: a prefix of $S_\tau$
          with $i \ge j_\tau + 1$, followed by
          $r \ge 0$ nodes whose frames were entered after $o_\tau$.
    \end{enumerate}
\end{lemma}
\begin{proof}
    The frames open at $o_\tau$ are the spine frames
    $(v_0, \dots, v_{l-1})$, each pushed with its spine search path, and
    by nesting they return in order of decreasing index. The frame of
    $w_\tau$ returns only at $\beta_\tau$, so throughout descent $\tau$
    the deepest open spine frame $v_i$ has $i \ge j_\tau + 1$, the frames
    below it are exactly $(v_0, \dots, v_{i-1})$ --- its search path ---
    and every frame above it was entered after $o_\tau$; this is (2).
    For (1): just after $\search(w_\tau)$ returns, the open frames are
    those below it, $(v_0, \dots, v_{j_\tau})$, which are the frames of
    the common prefix by the definition of the fork. In the degenerate
    case, the stack at $o_\tau = \beta_\tau$ is $(v_0, \dots, v_{l-1})$,
    which equals the common prefix since $j_\tau = l - 1$. At $\beta_0$
    no call has been entered, and at $\beta_T$ the initial call
    $\search(s) = \search(w_T)$ has returned, leaving no frame open ---
    the case $j_T = -1$, under which (2) reads $i \ge 0$.
\end{proof}

\begin{lemma}[Ascent Purity]\label{lem:phase-a-purity}
    No fruitful return, and hence no cascade, occurs in an ascent.
\end{lemma}
\begin{proof}
    Consider a return at a time $\theta \in (\beta_{\tau-1}, o_\tau)$. If
    the call was entered after $\beta_{\tau-1}$, its lifetime contains no
    output, since none lies in $(\beta_{\tau-1}, o_\tau)$; so the call is
    fruitless. If it was entered at or before $\beta_{\tau-1}$, its frame
    is open at $\beta_{\tau-1}$, hence by
    \autoref{lem:descent-stack}(1) it is one of
    the common-prefix frames of $S_{\tau-1}$ and $S_\tau$ --- for
    $\tau = 1$ and $\tau = T+1$, no frame is open at $\beta_{\tau-1}$ and
    this case is void. All
    common-prefix nodes lie on $S_\tau$, so their frames are on the stack
    at $o_\tau$ and do not return before $o_\tau$. Finally, no return
    occurs at $o_\tau$ itself: output is produced inside a for-loop.
\end{proof}

\begin{corollary}[Call--Phase Correspondence]\label{cor:call-phase}
    \hspace*{0pt} 
    \begin{enumerate}
    \item Every call opened during a descent returns fruitless within
          that descent.
    \item Every fruitless call lies within a single phase.
    \item Every fruitful call is opened in an ascent and returns in a
          descent; it is open at every output and every boundary
          $\beta_\sigma$ strictly inside its lifetime. In particular, a
          fruitful call spans at least the ascent and the descent
          surrounding its last output, and the initial call $\search(s)$
          is open from $\beta_0$ to $\beta_T$.
    \end{enumerate}
\end{corollary}
\begin{proof}
    (1) A frame open at $\beta_\tau$ is a common-prefix frame
    (\autoref{lem:descent-stack}(1)), which is open at $o_\tau$
    already and hence was not entered during descent $\tau$. So a call
    opened in descent $\tau$ returns within it; its lifetime contains no
    output, since no descent contains one; so it is fruitless.

    (2) A fruitless call is open at no output --- the frames open at
    $o_\tau$ lie on the output path and are fruitful --- and at no
    boundary $\beta_\tau$, by the prefix argument of (1). Its lifetime
    therefore crosses no phase boundary.

    (3) A fruitful call is open at some output, so by (1) it was not
    opened in a descent, and by \fullref{lem:phase-a-purity} it does not
    return in an ascent. Openness at intermediate boundaries follows
    since a call's lifetime is one contiguous span.
\end{proof}

\begin{lemma}[Descent Origins]\label{lem:spine-tail-origins}
    Every fruitful return in descent $\tau$ occurs at a spine frame of
    $S_\tau$: at a node $v_i \in S_\tau \setminus \{t\}$ whose live
    search path is $(v_0, \dots, v_i)$, with $i > j_\tau$.
    In particular, each node performs at most one fruitful return, and
    hence at most one origin dequeue, per descent.
\end{lemma}
\begin{proof}
    A call entered after $o_\tau$ that returns within descent $\tau$ is
    fruitless (\fullref{cor:call-phase}(1)). A call whose frame is open
    at $o_\tau$ is one of the spine frames $(v_0, \dots, v_{l-1})$, with
    live search path $(v_0, \dots, v_i)$. Of these, the frames with
    $i \le j_\tau$ are still open at $\beta_\tau$
    (\autoref{lem:descent-stack}(1)), so they do not return
    within the descent. A fruitful return in descent $\tau$ is thus a
    return of a spine frame with $i > j_\tau$; distinct spine frames sit
    at distinct nodes.
\end{proof}

\subsection{Attribution of Work}\label{sec:work-attribution}

Every elementary step of the execution belongs to a $\search$
call --- entry and exit bookkeeping, the successor scan of the for-loop,
and the post-loop assignment, in total at most $1 + |\suc(x)|$ steps per
call at $x$ --- or to a cascade --- dequeuing a pair at a node $x$ and
scanning its predecessors, at most $1 + |\pre(x)|$ steps per dequeue ---
or to producing an output path, at most $k + 1$ steps per output.
Each step is counted in the phase in which it is executed. By the
correspondence above, this attribution is unambiguous for all but one
kind of work:

\begin{itemize}
\item \emph{Fruitless calls.} All steps of a fruitless call lie in its
    single containing phase (\fullref{cor:call-phase}(2)).
\item \emph{Cascades.} A cascade runs atomically at a fruitful return,
    hence entirely within one descent (\fullref{lem:phase-a-purity},
    \fullref{lem:spine-tail-origins}) --- at most one per node per
    descent.
\item \emph{Outputs.} Producing $S_\tau$ is the final step of
    ascent $\tau$.
\item \emph{Fruitful calls.} Only here do the steps of one call spread
    over several phases. Each call examines every successor at most once
    during its entire lifetime, so its total work is at most
    $1 + |\suc(x)|$, and we charge a call its \emph{entire} scan cost once
    for each segment in which it executes a step --- a deliberate
    over-charge. For an \emph{interval} the charge is split: a call entered
    in it is counted by the per-phase clause of \fullref{cor:entry-bound},
    and a call already open at the left endpoint $o_\tau$ is a spine frame
    of $S_\tau$, the search path at $o_\tau$ being $S_\tau$ less $t$.
    For a \emph{period} that split is unavailable, the entry bound counting
    fruitless calls only; instead, \emph{every} fruitful call executing a
    step in period $\tau$ is a spine frame of $S_\tau$: it either returns
    within the period, in its descent and hence at a spine frame
    (\fullref{lem:spine-tail-origins}), or is still open at $\beta_\tau$ and
    is then a common-prefix frame (\autoref{lem:descent-stack}(1)), again a
    spine frame. Each segment below therefore carries exactly one spine
    account.
\end{itemize}

\subsection{Barrier Writes}

Consider the sequence of barrier assignments during the execution.
For a node $x$, we classify the assignments to $b[x]$:
\begin{itemize}
\item \emph{raise}:
    $b[x] \gets k - h + 1$, assigned in a fruitless $\search(x)$
    with search path of length $h$;
\item \emph{drop}:
    $b[x] \gets d + 1$, the predecessor write inside some
    cascade call $\fruitful(q, d)$ with $q \in \suc(x)$, under the guard
    $b[x] > d + 1$;
\item \emph{fruitful write}:
    $b[x] \gets sd$, the origin write at the start of a cascade
    $\fruitful(x, sd)$ invoked from a fruitful $\search(x)$.
\end{itemize}
There are no other types of assignments, and no other function call during
$\bsdfs(G, s, t, k)$ visits a node $x$ without assigning $b[x]$.
Note that the order of these writes, due to recursion, can differ from the
function-call order, since the barriers are written in postfix order
before return. Raises occur in both kinds of phases --- fruitless
excursions run during descents as well --- whereas drops and fruitful
writes are confined to descents, at most one fruitful write per node per
descent (\fullref{lem:phase-a-purity}, \fullref{lem:spine-tail-origins}).
Fruitful writes will play no role in the write counting; the following
corollary shows why.

\begin{corollary}[Fruitful Writes Never Lower]\label{cor:fruitful-never-lowers}
    At every fruitful return of $\search(v)$, the written value $sd$
    satisfies $sd \ge b[v]$ at write time. Hence, drops are the only
    writes that lower a barrier.
\end{corollary}
\begin{proof}
    In the actual execution, edge-consistency w.r.t.\ $P$ holds at the
    entry of every call (the joint induction of
    \autoref{sec:completeness}, cf.\
    \fullref{lem:search-preserves-consistency}), so
    \fullref{lem:per-call-monotone} applies to $\search(v)$:
    $b_{exit}[v] \ge b_{entry}[v]$. Between entry and the fruitful
    write, $b[v]$ is unwritten, because $v$ stays on the search path:
    on-path nodes are neither re-entered nor written by cascades; for
    the same reason, $b_{exit}[v] = sd$. Hence
    $sd = b_{exit}[v] \ge b_{entry}[v] = b[v]$ at write time.
\end{proof}

\begin{lemma}[Raise Bound]\label{lem:raise-bound}
    Consider, at a node $x \ne s$, a sequence of consecutive raises with no
    intervening drop at $x$. The sequence contains at most $k$ raises.
    At $x = s$, at most one raise occurs during the entire execution.
\end{lemma}
\begin{proof}
    By \fullref{lem:fruitless-increasing}, each raise at $x$ writes a value
    strictly exceeding the current $b[x]$. Between two consecutive raises of
    the sequence, $b[x]$ is modified only by fruitful writes, which never
    lower it (\fullref{cor:fruitful-never-lowers}); hence the written values
    are strictly increasing along the sequence. By
    \fullref{lem:barrier-ge-0}(2), each written value is at least $1$, and
    by \fullref{cor:barrier-upper}, $b[x] \le k$ for $x \ne s$. A strictly
    increasing chain of values in $\{1, 2, \dots, k\}$ has length at most
    $k$.
    For $x = s$, $\search$ is invoked at node $s$ exactly once during the
    entire execution (the initial call from $\bsdfs(G, s, t, k)$), so at
    most one raise at $s$ can occur.
\end{proof}

\begin{lemma}[Drop Bound]\label{lem:drop-bound}
    Consider, at a node $x$, a sequence of consecutive drops with no
    intervening raise or fruitful write at $x$. The sequence contains at
    most $k - 1$ drops. At $x = s$, no drop ever occurs.
\end{lemma}
\begin{proof}
    Each drop at $x$ writes $b[x] \gets d + 1$ under the guard
    $b[x] > d + 1$. The first drop of the sequence therefore writes a value
    strictly below the current $b[x] \le k$
    (\fullref{cor:barrier-upper}; the sentinel exception concerns only the
    final assignment at $s$, where no drop occurs), and each subsequent drop
    strictly decreases $b[x]$ further, since no other write intervenes. By
    \fullref{lem:barrier-ge-0}(2), each written value is at least $1$. A
    strictly decreasing chain of values in $\{1, 2, \dots, k-1\}$ has length
    at most $k - 1$.

    For $x = s$, the initial call $\bsdfs(G, s, t, k)$ pushes $s$ onto the
    search path and $s$ remains there throughout the execution; the
    $p \notin S$ guard in the cascade therefore prevents any drop at $s$.
\end{proof}

\begin{lemma}[Cascade Witness Path]\label{lem:cascade-witness}
    Under the cascade $\fruitful(v, sd)$ of \fullref{lem:cascade-distance}, invoked from a
    \emph{fruitful} $\search(v)$, every enqueued node $x$ admits an $x$-$t$-path of length $b[x]$
    whose nodes avoid $P$.
\end{lemma}
\begin{proof}
    By \fullref{lem:cascade-distance}, $b[x] = sd + dist_{Pv}(x,v)$.
    Take a shortest $x$-$v$-path realizing $dist_{Pv}(x,v)$; it meets $Pv$ only at $v$.
    Since $\search(v)$ is fruitful, it produces an output whose suffix from $v$ is a $v$-$t$-path
    of length $sd$ (the shortest produced), with interior nodes searched below $v$ and hence off $Pv$.
    Concatenating the two gives an $x$-$t$-path of length $dist_{Pv}(x,v) + sd = b[x]$
    whose only node in $Pv$ is the junction $v$; as $v \notin Pv \setminus \{v\}$, the path avoids $P$.
\end{proof}

The next lemma carries the witness-path argument that is the technical
core of this section: a descent never re-enters territory a cascade has
touched.

\begin{lemma}[Descent Seal]\label{lem:period-seal}
    In descent $\tau$, no $\search$ call is entered at a node
    that has received a drop since $o_\tau$.
\end{lemma}
\begin{proof}
    Suppose not, and let $\theta$ be the \emph{first} violating entry in
    this descent, at a node $x$.

    By \autoref{lem:descent-stack}(2), the search path at $\theta$ is
    $S' = (v_0, \dots, v_i, q_1, \dots, q_r = x)$; here
    $v_i$ is the deepest open spine frame and $i \ge j_\tau + 1$, and
    $r \ge 1$ because the frame entered at $\theta$ is not open at $o_\tau$
    and hence is not a spine frame.
    Write $h = \|S'\|$. By simplicity of the search path
    (Observation~\ref{obs:simple-search-path}), the segment $(q_1, \dots, x)$ has
    length $r - 1$ and avoids $\{v_0, \dots, v_i\}$.

    Consider the last drop at $x$ before $\theta$, necessarily after
    $o_\tau$, and let it write $b[x] = d$ inside a cascade
    $\fruitful(v_{i''}, sd'')$. By \fullref{lem:spine-tail-origins}, the
    origin of that cascade is a spine frame with live search path
    $(v_0, \dots, v_{i''})$. This origin has returned by $\theta$, whereas
    $v_i$ has not. Since spine frames return in order of decreasing index,
    $i'' > i$.

    No further write to $x$ occurs between this drop and $\theta$:
    \begin{itemize}
        \item A later drop at $x$ would contradict the choice of the last drop.
        \item A raise or a fruitful write at $x$ occurs only when a frame of
              $x$ returns. Such a frame cannot have been open at the time of
              the drop, since it would place $x$ on the search path, and
              cascades do not write to search-path nodes. It would therefore
              have to be entered after the drop and before $\theta$. But that
              entry is itself at a node which has received a drop since
              $o_\tau$, contradicting the choice of $\theta$ as the first
              such entry.
    \end{itemize}
    Hence $b[x] = d$ still holds at $\theta$. By
    \fullref{lem:cascade-witness}, there is an $x$-$t$-path $W$ of length
    $d$ avoiding $\{v_0, \dots, v_{i''-1}\} \supseteq \{v_0, \dots, v_i\}$.

    The pruning guard at $\theta$ gives $b[x] \le k - h$, so $d \le k - h$,
    that is, $h + d \le k$.

    Concatenating $(q_1, \dots, x)$ with $W$ yields a $q_1$-$t$-path $Q$ in
    $G \setminus \{v_0, \dots, v_i\}$ of length $(r - 1) + d$. Now $Q$ need
    not be simple, but shortcutting any repeated nodes gives a simple
    $q_1$-$t$-path $Q'$ in the same sub-graph, of length at most
    $(r - 1) + d$. Prepending $(v_0, \dots, v_i, q_1)$ extends $Q'$ to a
    simple $s$-$t$-path $R$ of length at most
    $(i + 1) + (r - 1) + d = h + d \le k$, whose first $i + 2$ nodes
    coincide with the current search path truncated at $q_1$.

    By \fullref{thm:bsdfs-complete}, $R$ is produced by
    $\bsdfs(G, s, t, k)$, and by
    Observation~\ref{obs:search-path-uniqueness} each search path occurs at
    most once during the execution. The search path
    $(v_0, \dots, v_i, q_1)$ is currently active, so $R$ is produced during
    the ongoing call $\search(q_1)$.

    We claim that this call was entered after $o_\tau$. First,
    $q_1 \ne v_{i+1}$: otherwise $(v_0, \dots, v_i, q_1)$ would be the
    search path of the spine frame of $v_{i+1}$, which has already been
    popped because $v_i$ is the deepest open spine frame, and this search
    path would occur twice, contradicting
    Observation~\ref{obs:search-path-uniqueness}; for $i = l - 1$ we have
    $v_{i+1} = t$, which never receives a frame, so $q_1 \ne v_{i+1}$ holds
    outright. Consequently
    $(v_0, \dots, v_i, q_1)$ is not a spine search path, so the frame of
    $q_1$ is not a spine frame. Every frame open at $o_\tau$ is a spine
    frame, and the claim follows.

    No descent contains an output, so no output falls between $o_\tau$ and
    the entry of $\search(q_1)$. Hence the first output after $o_\tau$
    falls inside $\search(q_1)$.

    For $\tau = T$ this is already a contradiction, since $o_T$ is the last
    output. For $\tau < T$, that first output is $o_{\tau+1}$, so
    $S_{\tau+1}$ begins with $(v_0, \dots, v_i, q_1)$. Since
    $q_1 \ne v_{i+1}$, the fork of $S_\tau$ and $S_{\tau+1}$ is $v_i$,
    giving $j_\tau = i \ge j_\tau + 1$ --- a contradiction.
\end{proof}

\subsection{Intervals and Periods}\label{sec:periods}

Two coarser decompositions arise by pairing adjacent phases, see \autoref{fig:timeline}.

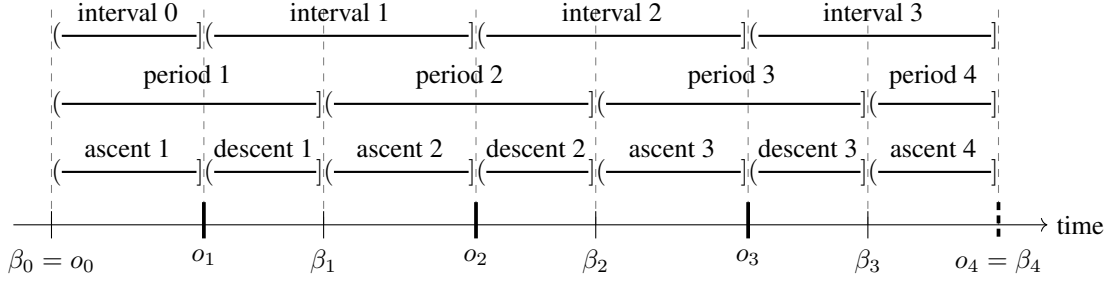
\begin{figure}[h]
\centering
\def\hospan#1#2#3#4{%
    \node[brk, anchor=west] (spanL) at (#1,#3) {$($};
    \node[brk, anchor=east] (spanR) at (#2,#3) {$]$};
    \draw[thick] (spanL.east) -- (spanR.west);
    \node[lab, above=2pt] at ({(#1+#2)/2},#3) {#4};
}
\begin{tikzpicture}[x=0.72cm, y=1cm,
    brk/.style = {inner sep=0pt, outer sep=0pt},
    guide/.style     = {gray, dashed},
    lab/.style       = {}]

  \def\bZero{0}   \def\oOne{2.8}   \def\bOne{5.0}   \def\oTwo{7.8}
  \def\bTwo{10.0} \def\oThree{12.8} \def\bThree{15.0} \def\oFour{17.4}

  \foreach \x in {\bZero,\oOne,\bOne,\oTwo,\bTwo,\oThree,\bThree,\oFour}
      \draw[guide] (\x,0.10) -- (\x,3.05);

  \draw[->] (-0.7,0.10) -- (18.3,0.10) node[right] {time};
  \foreach \x in {\bZero,\bOne,\bTwo,\bThree}
      \draw (\x,0.28) -- (\x,-0.08);                      
  \foreach \x in {\oOne,\oTwo,\oThree}
      \draw[line width=1.4pt] (\x,0.40) -- (\x,-0.08);    
  \draw[line width=1.4pt, densely dashed] (\oFour,0.40) -- (\oFour,-0.08);

  \node[below] at (\bZero,-0.10)  {$\beta_0 = o_0$};
  \node[below] at (\oOne,-0.10)   {$o_1$};
  \node[below] at (\bOne,-0.10)   {$\beta_1$};
  \node[below] at (\oTwo,-0.10)   {$o_2$};
  \node[below] at (\bTwo,-0.10)   {$\beta_2$};
  \node[below] at (\oThree,-0.10) {$o_3$};
  \node[below] at (\bThree,-0.10) {$\beta_3$};
  \node[below] at (\oFour,-0.10)  {$o_4 = \beta_4$};

  \hospan{\bZero}{\oOne}{0.80}{ascent 1}
  \hospan{\oOne}{\bOne}{0.80}{descent 1}
  \hospan{\bOne}{\oTwo}{0.80}{ascent 2}
  \hospan{\oTwo}{\bTwo}{0.80}{descent 2}
  \hospan{\bTwo}{\oThree}{0.80}{ascent 3}
  \hospan{\oThree}{\bThree}{0.80}{descent 3}
  \hospan{\bThree}{\oFour}{0.80}{ascent 4}

  \hospan{\bZero}{\bOne}{1.70}{period 1}
  \hospan{\bOne}{\bTwo}{1.70}{period 2}
  \hospan{\bTwo}{\bThree}{1.70}{period 3}
  \hospan{\bThree}{\oFour}{1.70}{period 4}

  \hospan{\bZero}{\oOne}{2.60}{interval 0}
  \hospan{\oOne}{\oTwo}{2.60}{interval 1}
  \hospan{\oTwo}{\oThree}{2.60}{interval 2}
  \hospan{\oThree}{\oFour}{2.60}{interval 3}

\end{tikzpicture}
\caption{Decomposition of an execution with $T = 3$ outputs. Bold ticks mark
the outputs $o_1, \dots, o_T$, the dashed tick the terminal event
$o_{T+1}$, thin ticks the boundaries $\beta_\tau$, with $\beta_0 = o_0$ the
start and $o_{T+1} = \beta_{T+1}$ the termination of the execution.
All segments are left half-open. Periods are cut at the
$\beta_\tau$, intervals at the $o_\tau$; hence each period contains exactly
one event and each interval carries its event at the right end. A descent
may be empty ($o_\tau = \beta_\tau$), and descent $T+1$ always is; the
others are drawn non-empty here. The terminal period carries the terminal
event only; it is drawn wide for legibility.}
\label{fig:timeline}
\end{figure}

The events cut the execution into $T + 1$ \emph{intervals}, indexed
$\tau = 0, \dots, T$: interval $\tau$ is $(o_\tau, o_{\tau+1}]$, that is,
descent $\tau$ followed by ascent $\tau + 1$; interval $0$ is ascent $1$
alone, descent $0$ being empty. Interval $\tau$ therefore opens immediately
after $o_\tau$, closes with the event $o_{\tau+1}$, and carries no event in
its interior. The \emph{delay} of an interval is the number of elementary
steps executed in it --- including those producing the closing output ---
and the worst-case delay of the algorithm is the maximum over all
intervals. For $T = 0$ there is the single interval $0$, the entire
execution.

\emph{Period} $\tau$, for
$1 \le \tau \le T+1$, is the segment $(\beta_{\tau-1}, \beta_\tau]$,
ascent $\tau$ followed by descent $\tau$. By the boundary ordering, the
periods tile the execution and period $\tau$ contains exactly the event
$o_\tau$. There are thus $T+1$ periods and $T+1$ intervals, one for each event after the start.
An interval straddles two periods, and vice versa, and this misalignment is
where the higher constant of the worst-case bound comes from.

Note the asymmetry in the indexing: period $\tau$ contains the event
$o_\tau$, whereas interval $\tau$ is opened by $o_\tau$ and contains
$o_{\tau+1}$. Both consist of two phases and both inherit one spine account
from the frames open at their left endpoint. They differ in the entries:
per node, a period admits a single chain of fruitless entries, an interval
two (\fullref{cor:entry-bound}). That one extra chain is the entire
difference between the worst-case and the amortized constant below.

Using the phase lemmas, we can prove a central property of a phase.

\begin{lemma}[One Entry per Pair per Phase]\label{lem:one-entry-phase}
    Let $\Phi$ be an ascent or a descent. For each node $x$ and each depth
    $h$, at most one call $\search(x)$ with search path of length $h$ is
    entered in $\Phi$.
\end{lemma}
\begin{proof}
    Suppose two such calls are entered in $\Phi$, at times
    $\theta_1 < \theta_2$. While a frame of $x$ is open, $x$ lies on the
    search path (Observation~\ref{obs:simple-search-path}), so the first
    frame has returned before $\theta_2$, hence within $\Phi$. Its return is
    fruitless: in a descent by \fullref{cor:call-phase}(1), in an ascent
    because a fruitful call returns in a descent
    (\fullref{lem:phase-a-purity}). The return is therefore a raise, writing
    $b[x] = k - h + 1$, whereas the pruning guard at $\theta_2$ requires
    $b[x] \le k - h$. Since only drops lower barriers
    (\fullref{lem:fruitless-increasing},
    \fullref{cor:fruitful-never-lowers}), a drop at $x$ occurs between that
    return and $\theta_2$, hence inside $\Phi$.

    If $\Phi$ is an ascent, it contains no cascade and hence no drop
    (\fullref{lem:phase-a-purity}) --- a contradiction. If $\Phi$ is
    descent $\tau$, the drop occurs after $o_\tau$, so $\theta_2$ is an
    entry at a node that has received a drop since $o_\tau$, contradicting
    \fullref{lem:period-seal}. (For $x = s$ the claim is trivial:
    $\search(s)$ is entered once per execution.)
\end{proof}

\begin{corollary}[Two Entries per Pair]\label{cor:two-entries}
    In every interval, for each node $x$ and each depth $h$, at most two
    calls $\search(x)$ with search path of length $h$ are entered.
\end{corollary}
\begin{proof}
    An interval consists of two phases --- descent $\tau$ followed by
    ascent $\tau+1$ --- and \fullref{lem:one-entry-phase} permits one entry
    per pair in each. In the initial interval the descent is empty, and in
    the terminal one the ascent carries no call, so both admit at most one.
\end{proof}

The constant $2$ counts phases, not calls, and it is not an artifact of the
analysis: Johnson's Lemma~3 --- at most two \textsc{Circuit} calls per
vertex between consecutive
outputs~\cite{DBLP:journals/siamcomp/Johnson75} --- is the $k \to \infty$
shadow of \fullref{cor:two-entries} at the same constant, the barrier
collapsing to a binary blocked flag and the node-depth pair to the node.

Counting entries per node rather than per pair turns the two statements
above into the form the delay proofs use.

\begin{corollary}[Entry Bound]\label{cor:entry-bound}
    Let $x \ne s$. In any phase, at most $k$ calls at $x$ are entered; in
    any period, at most $k$ \emph{fruitless} calls at $x$ are entered. The
    call $\search(s)$ is entered once during the entire execution. Charging
    every call its entire scan cost, the calls counted in either case
    account for at most
    $\sum_{x \in V} k\,(1 + |\suc(x)|) = k\,(n+m)$ steps.
\end{corollary}
\begin{proof}
    A call at $x \ne s$ is entered at a depth $h$ with $1 \le h \le k$: the
    pruning guard gives $b[x] + h - 1 < k$, and $b[x] \ge 0$
    (\fullref{lem:barrier-ge-0}). It therefore suffices that the entries
    counted have pairwise distinct depths, which for a phase is
    \fullref{lem:one-entry-phase}.

    For a period, each fruitless call at $x$ ends in exactly one raise, and
    the raises at $x$ within the period form a single sequence with no
    intervening drop: the ascent contains no drops
    (\fullref{lem:phase-a-purity}); a raise in the descent is the return of
    a call entered in the descent (\fullref{cor:call-phase}(1),
    \fullref{lem:spine-tail-origins}), so by \fullref{lem:period-seal} its
    entry --- and, since cascades do not write on-path nodes, its return ---
    precedes $x$'s first drop of the period; and the drops of the period all
    lie in the descent. By \fullref{lem:raise-bound} such a sequence has at
    most $k$ raises. A call at $x$ costs at most $1 + |\suc(x)|$ steps over
    its whole lifetime (\autoref{sec:work-attribution}).
\end{proof}

\subsection{Delay Theorems}\label{sec:delay-theorems}

\begin{theorem}[Worst-Case Delay]\label{thm:worst-case-delay}
    For a graph $G$ with $n$ nodes and $m$ edges, and the standing length
    bound $0 < k \le n$, $\bsdfs(G,s,t,k)$ enumerates with a worst-case
    delay of at most $3(k+1)(n+m) \in O(k(n+m))$ elementary
    steps between consecutive events $o_0, o_1, \dots, o_{T+1}$.
\end{theorem}
\begin{proof}
    The segments between consecutive such events are exactly the
    intervals.
    Fix interval $\tau$ and tally the steps executed in it, following
    \autoref{sec:work-attribution}.

    \begin{description}
    \item[Search work, calls entered in the interval.]
        Interval $\tau$ consists of the two phases descent $\tau$ and
        ascent $\tau+1$, one of which is void in the initial and in the
        terminal interval. By \fullref{cor:entry-bound}, at most $k$ calls per node are
        entered in each phase, and charging every one of them its entire
        scan cost:
        $$\sum_{x \in V} 2k \cdot (1 + |\suc(x)|) \;=\; 2k\,(n+m).$$
    \item[Search work, calls open at $o_\tau$.]
        A call executing a step in the interval without being entered in it
        is open at $o_\tau$, since a call's lifetime is one contiguous
        span. The frames open at $o_\tau$ are the spine frames of
        $S_\tau$: the search path at $o_\tau$ is $S_\tau$ less $t$. They
        sit at pairwise distinct nodes of $S_\tau \setminus \{t\}$ --- at most $k$
        of them (\fullref{lem:search-path-length}). Charging each its entire
        scan cost:
        $$\sum_{v_i \in S_\tau \setminus \{t\}}
            \bigl(1 + |\suc(v_i)|\bigr) \;\le\; k + m.$$
        For the initial interval this account is empty: $S_0$ is the empty
        spine and no frame is open at $o_0 = \beta_0$
        (\autoref{lem:descent-stack}(1)).
    \item[Drop work.]
        All dequeues executed in interval $\tau$ lie in descent $\tau$,
        since ascent $\tau+1$ contains no cascade
        (\fullref{lem:phase-a-purity}). Per node, they comprise at most one
        origin dequeue (\fullref{lem:spine-tail-origins}) and the dequeues
        of the drops of the descent, which form a single sequence with no
        intervening raise or fruitful write --- by \fullref{lem:period-seal}
        no entry, hence no raise and no open frame of $x$, occurs after
        $x$'s first drop, and no frame of $x$ is open at that first drop ---
        at most $k - 1$ (\fullref{lem:drop-bound}); at most $k$ dequeues in
        total, each costing $1 + |\pre(x)|$:
        $$\sum_{x \in V} k \cdot (1 + |\pre(x)|) \;=\; k\,(n+m).$$
    \item[Output work.]
        Interval $\tau$ closes with the event $o_{\tau+1}$; producing an
        output costs at most $k + 1$ steps, and the terminal event
        $o_{T+1}$ costs none.
    \end{description}

    Summing, the per-interval delay is bounded by
    $$2k(n+m) + (k + m) + k(n+m) + (k + 1) \;\le\; 3(k+1)(n+m),$$
    the last inequality using $k \le n$ and $m \ge 1$; for $m = 0$ the run
    is the single fruitless call $\search(s)$ and the bound is trivial.
\end{proof}

For the initial and the terminal interval we can derive tighter bounds.

\begin{corollary}[Boundary Intervals]\label{cor:boundary-intervals}
    The initial interval has delay at most $k(n+m) + (k+1) \le
    (k+1)(n+m)$, and at most $k(n+m)$ when no output is produced.
    The terminal interval has delay at most $2(k+1)(n+m)$.
\end{corollary}
\begin{proof}
    The initial interval coincides with the first ascent (or with the
    entire run, if no output is produced). It carries one phase-budget of
    entries --- at most $k(n+m)$ scan steps (\fullref{cor:entry-bound})
    --- no spine frames (\fullref{lem:descent-stack}(1)) and, containing
    no cascade (\fullref{lem:phase-a-purity}), no dequeues; the closing
    output, if any, costs at most $k+1$. The terminal interval, being
    descent $T$ followed by the void ascent $T+1$, carries one
    phase-budget of entries, the spine frames of $S_T$ open at $o_T$ ---
    at most $k+m$ scan steps --- its dequeues --- at most $k(n+m)$ ---
    and no output:
    $$k(n+m) + (k+m) + k(n+m) \;\le\; 2(k+1)(n+m),$$
    both final inequalities using $k \le n$ and $m \ge 1$.
\end{proof}

The accounts in \fullref{thm:worst-case-delay} are of different kinds,
and only the first is doubled by the interval:
two phase-budgets of entries, one spine, one descent's worth of dequeues.
Per-node accounting cannot do better than this.
The two entry chains of an interval can be nearly full simultaneously --- they are the
chains of two adjacent periods, separated by a single drop --- and entries
at a node do not trade off against its dequeues: after $r$ entries in a
phase the barrier at $x$ is at least $r$, leaving room for \emph{more}
drops, not fewer.

At period granularity the entry account halves. The raises of a period ---
those of its ascent together with those of its descent --- form a single
chain with no intervening drop, so per node one chain of fruitless entries
replaces two (\fullref{cor:entry-bound}); the spine, drop and output
accounts are unchanged, the period being closed by the boundary
$\beta_\tau$ and containing exactly one event, namely $o_\tau$.
As the periods correspond one-to-one to the events $o_1, \dots, o_{T+1}$,
the per-period tally below is a per-event tally.

\begin{theorem}[Amortized Delay]\label{thm:amortized-delay}
    For a graph $G$ with $n$ nodes and $m$ edges, and the standing length
    bound $0 < k \le n$, $\bsdfs(G,s,t,k)$ reaches event $o_p$
    within $2p(k+1)(n+m)$ elementary steps, for $p = 1, \dots, T+1$.
    In particular, it produces its first $p$ outputs within $2p(k+1)(n+m)$ steps,
    so the amortized per-output delay is at most $2(k+1)(n+m) \in O(k(n+m))$.
\end{theorem}
\begin{proof}
    Fix a period $\tau$, $1 \le \tau \le T+1$, and tally the steps executed
    in it, following \autoref{sec:work-attribution}.

    \begin{description}
    \item[Search work, fruitless calls.]
        Each fruitless call lies within one phase
        (\fullref{cor:call-phase}(2)), hence within one period. By
        \fullref{cor:entry-bound}, at most $k$ of them are entered per node
        per period:
        $$\sum_{x \in V} k \cdot (1 + |\suc(x)|) \;=\; k\,(n+m).$$
    \item[Search work, fruitful calls.]
        Every fruitful call executing a step in period $\tau$ either returns
        within the period --- in the descent, at a spine frame of $S_\tau$
        (\fullref{lem:phase-a-purity},
        \fullref{lem:spine-tail-origins}) --- or is open at the boundary
        $\beta_\tau$, hence a common-prefix frame
        (\autoref{lem:descent-stack}(1)), also a spine frame of $S_\tau$.
        Either way, these are pairwise distinct frames at pairwise distinct
        nodes on $S_\tau \setminus \{t\}$ --- at most $k$ of them
        (\fullref{lem:search-path-length}). Charging each its entire scan
        cost:
        $$\sum_{v_i \in S_\tau \setminus \{t\}}
            \bigl(1 + |\suc(v_i)|\bigr) \;\le\; k + m.$$
        For $\tau = T+1$ the account is void: $S_{T+1}$ is the empty spine,
        and no fruitful call executes a step in the terminal period ---
        for $T \ge 1$ because no frame is open at $\beta_T$
        (\autoref{lem:descent-stack}(1)), and for $T = 0$ because no call
        is fruitful at all.
    \item[Drop work.]
        All dequeues executed in period $\tau$ lie in descent $\tau$
        (\fullref{lem:phase-a-purity}), and the argument of
        \fullref{thm:worst-case-delay} applies verbatim --- the drops of a
        period all lie in its descent --- giving at most $k$ dequeues per
        node:
        $$\sum_{x \in V} k \cdot (1 + |\pre(x)|) \;=\; k\,(n+m).$$
    \item[Output work.]
        Period $\tau$ contains exactly the event $o_\tau$; producing an
        output costs at most $k + 1$ steps, and the terminal event costs
        none.
    \end{description}

    Summing, the per-period tally is bounded by
    $$k(n+m) + (k + m) + k(n+m) + (k + 1) \;\le\; 2(k+1)(n+m),$$
    the last inequality using $k \le n$ and $m \ge 1$; for $m = 0$ the run is
    the single fruitless call $\search(s)$ and the bound is trivial.

    Summing up over the first $p$ periods, one for each of $o_1, \dots, o_p$,
    yields the desired bound.
\end{proof}

The $O(k(n+m))$ worst-case and amortized delay bounds carry over to the
enumeration of all simple cycles of length at most $k$
after the initial $O(n(n+m))$ pre-processing of \autoref{sec:all-cycles}.
The runs $\bsdfs(G^i, v_i, v_i, k)$ for $v_i \in L$ output every such cycle
exactly once, and the delay bound arguments apply to $G^i$.
Terminating the current run costs $O(k(n+m))$ and setting up the next one
$O(n+m)$; since every executed run produces at least one output, both are
absorbed.

We close this section with a graph family on which $\bsdfs$ exhibits
$\Omega(k(n+m))$ delay.

\begin{lemma}[Clique Traps]\label{lem:clique-trap}
    Let $N_{c,k} := kc - k(k+1)/2 + 1$.
    \begin{enumerate}
    \item For $0 < k < c$ let $G_c$ be the digraph on
        $\{0,\dots,c-1\} \cup \{t\}$ whose edges are all ordered pairs of
        distinct nodes below $c$, so that $t$ is isolated and $n = c+1$,
        $m = c(c-1)$. Then, for every successor scan order,
        $\bsdfs(G_c,0,t,k)$ produces no output, and its single interval
        has delay exactly $c\,N_{c,k}$.
    \item For $3 \le k < c$ let $G_c^{+}$ arise from $G_c$ by adding the
        edge $(0,t)$, scanned first at node $0$; then $n = c+1$ and
        $m = c(c-1)+1$. Then $\bsdfs(G_c^{+},0,t,k)$ produces the single
        output path $(0,t)$ after $4$ steps, and its terminal interval has delay
        exactly $c\,N_{c,k} + c^{2} - 1$.
    \end{enumerate}
\end{lemma}

\begin{proof}
    (1) No node reaches $t$, so no call returns fruitfully: every barrier
    is written by a raise, no cascade is executed, and there is no output.
    A call at $v$ entered at depth $h$ passed $b[v] + (h-1) < k$ and writes
    $b[v] = k-h+1$ on return, so any later call at $v$ has depth $h' < h$;
    entry depths of a node strictly decrease.

    We claim that for $1 \le h \le k$ the calls at depth $h$ are exactly
    the $c-h$ children of the \emph{first} call at depth $h-1$. That call
    is entered during the initial descent, so its stack is
    $\{0,\dots,h-1\}$ and no node outside has been entered. Of its $c-1$
    successors, $h-1$ lie on the stack. An off-stack successor $w$ was, if
    entered at all, entered inside the subtree of an earlier off-stack
    successor, hence at depth $\ge h+1$, so $b[w] \le k-h$ and the guard
    $b[w] + (h-1) < k$ holds; all $c-h$ of them are therefore entered at
    depth $h$. Afterwards each carries $b = k-h+1$, which fails that same
    guard, so no later call reaches depth $h$. Summing,
    $\sum_{h=1}^{k}(c-h) + 1 = N_{c,k}$ calls, each scanning $c-1$
    successors at cost $1 + (c-1) = c$ (\autoref{sec:work-attribution}),
    while drop and output work are void.

    (2) In $G_c^{+}$ the scan of $t$ at node $0$ passes $b[t] + 0 < k$ and
    sets $sd \gets 1$; it enters no call and writes no barrier, so the call
    structure and all raises are those of (1) verbatim, leaving $b[j] = k$
    for $1 \le j < c$. Node $0$ scans one successor more, giving
    $c\,N_{c,k} + 1$ scan steps, and the output $(0,t)$ costs $2$. Since
    $sd = 1 \le k$, $\search(0)$ returns fruitfully: the cascade sets
    $b[0] \gets 1$ and dequeues $0$ at cost $1 + |\pre(0)| = c$; as
    $k \ge 3$ we have $b[j] = k > 2$, so every $b[j]$ drops to $2$ and is
    dequeued at cost $1 + |\pre(j)| = c$, giving $c^{2}$ in total, after
    which no barrier exceeds $2$ and the cascade stops. The interval
    boundary lies after the output, at step $4$.
\end{proof}

In \fullref{lem:clique-trap}, we have $n+m = c^{2}+1$ resp.\ $c^{2}+2$.
For $c \to \infty$ and $k \in o(n)$, e.g.\ fixed $k$, the run on $G_c$ is
a single fruitless initial interval, for which
\fullref{cor:boundary-intervals} proves the bound $k(n+m)$; the delay
$c\,N_{c,k}$ divided by $k(n+m)$ tends to $1$, so that bound is
asymptotically tight: the trap saturates the entry account of
\fullref{cor:entry-bound} node for node, and pays nothing for drops or
outputs. Against the general theorems the same run leaves the factors
$2(k+1)/k$ resp.\ $3(k+1)/k$ --- slack residing entirely in the accounts
the trap does not exercise. The one added edge of $G_c^{+}$ supplies part
of what the trap lacks: $\search(0)$ returns fruitfully, and the
terminating cascade sweeps the clique at cost $1 + |\pre(x)|$ per node, a
further $c^{2} \sim n+m$ steps. Its terminal-interval delay divided by
$(k+1)(n+m)$ therefore tends to $1$, bounding
\autoref{thm:worst-case-delay} from below by a factor $3$ and the
terminal-interval bound of \fullref{cor:boundary-intervals} by a factor
$2$: the cascade dequeues each node once, where the budget admits $k$
dequeues per node. We are not aware of a family forcing more.

\section{Edge-Consistent Variants of $\bsdfs$}\label{sec:bsdfs-variants}

The pruning of $\bsdfs$ is governed entirely by the barrier values, and the only
freedom in the algorithm is \emph{how} the algorithm revises the barriers.
Several schemes share the same skeleton --- the recursive search, and the pruning condition
$b[w] + h < k$ --- and differ only in the barrier post-processing after recursion.
We call a scheme \emph{edge-consistent} if it maintains edge-consistency with respect to the current search path
at every \emph{quiescent point}, i.e.\ whenever control is in $\search$ and not inside a fruitful cascade;
by \fullref{lem:pruning-is-permissive} every edge-consistent scheme is complete.
The \emph{tight} scheme of the original $\bsdfs$ is one extreme.
We record three further points on the spectrum.

\subsection{The Trivial Scheme}\label{sec:trivial-scheme}
Leave every barrier at $0$: the algorithm does not manage barriers at all.
Then $b \equiv 0$ is trivially edge-consistent w.r.t.\ any search path, so the scheme is complete ---
it is plain depth-first search with the length cutoff $h < k$, performing no barrier-based pruning.
Soundness and completeness obviously hold, but there is no polynomial delay bound:
if the depth-first search first descends into a complete subgraph of size $k-1$ not
containing the target $t$, it traverses $\Theta((k-2)!)$ partial paths before
producing any output.

\subsection{The Loose Scheme}\label{sec:loose-scheme}

The \emph{loose} scheme keeps the backward cascade of the fruitful case but replaces
the distance values it deposits by $0$. That is, \fullref{alg:fruitful} is replaced by
$\textsc{Reset}(v)$: a breadth-first traversal backwards from $v$ over predecessors,
restricted to nodes off the search path, setting every visited barrier to $0$.

\begin{algorithm}[H]
    \caption{$\textsc{Reset}$ (loose fruitful case)}\label{alg:reset}
    \begin{algorithmic}[1]
        \STATE \textbf{$\textsc{Reset}(v)$}
        \STATE \quad $b[v] \gets 0$
        \STATE \quad $Q \gets$ queue containing $v$
        \STATE \quad \textbf{while} $Q$ not empty \textbf{do}
        \STATE \quad\quad $q \gets$ dequeue $Q$
        \STATE \quad\quad \textbf{for each} $p \in \pre(q)$ \textbf{do}
        \STATE \quad\quad\quad \textbf{if} $p \notin S$ \textbf{and} $b[p] \ne 0$ \textbf{then}
        \STATE \quad\quad\quad\quad $b[p] \gets 0$
        \STATE \quad\quad\quad\quad enqueue $p$ onto $Q$
    \end{algorithmic}
\end{algorithm}

The guard $b[p] \ne 0$ cannot be tightened to the depth-aware $b[p] > h + 1$ of the
tight cascade. Because $\textsc{Reset}$ records no distance information, edge-consistency
is preserved only by resetting the \emph{entire} reachable off-path cone: stopping the
propagation early would leave a node with a positive barrier adjacent to a node reset to
$0$, breaking consistency. The full cone is the price of discarding distances, and it is
why the loose scheme has no delay advantage --- a point we return to below.

The conceptual value of the loose scheme is that it isolates which part of the tight
correctness proof is essential. We show it is edge-consistent; completeness then follows
from \fullref{lem:pruning-is-permissive} as for any edge-consistent scheme. As in the
tight case, a cascade transiently violates edge-consistency while in flight and restores
it on completion, so the statements below concern quiescent points.

The case that dominates the tight proof disappears. A fruitful node ends its cascade
with $b[v] = 0$ (line 2 of \autoref{alg:reset}), so every \emph{outgoing} edge $(v,y)$
satisfies $b[v] \le b[y] + 1$ as $0 \le b[y] + 1$. This is precisely the case for which
the tight proof needs the exact value $sd = dist_{Pv}(v,t)$, and hence completeness;
here it is free. The only case requiring an argument is an off-path edge $(x,y)$ whose
head $y$ was reset to $0$, settled by the following closure property.

\begin{lemma}[Reset Closure]\label{lem:reset-closure}
    Let $\textsc{Reset}(v)$ run with search path $S$, and let $R$ be the set of nodes it
    sets to $0$. Then at the cascade's completion, every off-path predecessor of a node
    in $R$ has barrier $0$.
\end{lemma}
\begin{proof}
    Every node placed in $R$ is enqueued, and the while-loop runs until $Q$ is empty, so
    each $y \in R$ is dequeued and \emph{all} its predecessors are examined (lines 6--8).
    For an off-path predecessor $x$ of such a $y$: if $b[x] \ne 0$ the guard sets it to
    $0$; otherwise $b[x] = 0$ already. Either way $b[x] = 0$, and since the cascade only
    sets values to $0$, this persists to completion.
\end{proof}

If $(x,y)$ is off-path with $y \in R$, then $x$ is an off-path predecessor of a reset
node, so $b[x] = 0$ and $b[x] \le b[y] + 1$. The configuration one might fear --- a
fruitless node $x$ with a high barrier pointing into a freshly-reset $y$ --- cannot
arise, because the cascade reaches $x$ \emph{through} $y$ and resets it too. Off-path
edges with both endpoints untouched inherit consistency from before the cascade, and the
edges that the pop of $v$ newly constrains are consistent for the same two reasons
($b[v] = 0$ outgoing, \autoref{lem:reset-closure} with $y = v$ incoming). The fruitless
case and the push are as in the tight scheme. Hence, the loose scheme is edge-consistent,
and the four-way joint induction of \autoref{sec:completeness} collapses to a single
inequality invariant: no exact-value claim, no coupling to completeness.

\paragraph{Delay.}
Both tight and loose schemes are complete and deterministic and scan successors in the same order.
Since pruning removes only output-free subtrees, each emits its outputs in
lexicographic path order; hence the two output sequences are equal.
Therefore, the intervals correspond one-to-one and the spine paths coincide.
The fruitful calls returning within an interval are exactly the calls along its spine,
so the spine account of the tally in \fullref{thm:worst-case-delay} is identical in
both schemes.

Nonetheless, fruitless searches within a phase can multiply: barriers
are reset to $0$ instead of being updated to the tight bound,
\fullref{lem:one-entry-phase} fails for the loose scheme, and pruned
subtrees can be re-explored.
%
%
This is not merely a gap in the analysis.
In \autoref{app:loose-breaker} we construct a family $(G_k)$ with
$n + m = \Theta(k)$ on which the loose scheme performs $\Omega(k^2(n+m))$
steps between two consecutive outputs (\autoref{thm:loose-breaker}).
Hence, the loose scheme admits no $O(k(n+m))$ delay bound: the distance information
deposited by the tight cascade is essential to \fullref{thm:worst-case-delay}.

\subsection{The Lazy Scheme}\label{sec:lazy-scheme}

The loose scheme (\autoref{sec:loose-scheme}) discards distance information
but still incurs, at every fruitful step, the cost of a backward cascade
over the entire reachable off-path cone.
The \emph{lazy} scheme retains the same coarse reset to $0$ while deferring the cascade.
A fruitless node \emph{registers} itself in a dependency list at each of its successors,
and a later fruitful node \emph{consumes} those lists, clearing only the nodes that registered through it.
Distance information is still discarded---a fruitful node is reset to $0$ exactly as in the loose scheme---
so the lazy scheme inherits the loose scheme's simple correctness argument,
which does not rely on exact distance values.
The only difference is \emph{which} barriers are cleared by a fruitful step,
and \emph{when} they are cleared.

Each node $x$ now carries, in addition to its barrier $b[x]$, a \emph{dependency list} $B[x] \subseteq V$.
When a call $\search(v)$ returns fruitless at depth $h$, it performs the usual assignment
\[
b[v] \gets k-h+1
\]
and additionally inserts $v$ into $B[w]$ for every successor $w \in \suc(v)$ (the registration step, line~\ref{line:register}).
The fruitful case runs $\textsc{Update}(v)$ instead of the loose reset;
both are given in \autoref{alg:update-lazy}.

\begin{algorithm}[H]
    \caption{Lazy fruitful and fruitless handling, replacing the body of the
    fruitful/fruitless branch in $\search$ (\autoref{alg:search}).}\label{alg:update-lazy}
    \begin{algorithmic}[1]
        \STATE \textbf{if} $sd \le k$ \textbf{then} \quad \# fruitful
        \STATE \quad $b[v] \gets 0$
        \STATE \quad $\textsc{Update}(v)$
        \STATE \textbf{else} \quad \# fruitless
        \STATE \quad $b[v] \gets k - h + 1$
        \STATE \quad \textbf{for each} $w \in \suc(v)$ \textbf{do} \quad $B[w] \gets B[w] \cup \{v\}$ \label{line:register}
        \STATE
        \STATE \textbf{$\textsc{Update}(v)$}
        \STATE \quad \textbf{for each} $u \in B[v]$ \textbf{do}
        \STATE \quad\quad \textbf{if} $u \notin S$ \textbf{and} $b[u] > 0$ \textbf{then}
        \STATE \quad\quad\quad $b[u] \gets 0$
        \STATE \quad\quad\quad $\textsc{Update}(u)$
        \STATE \quad $B[v] \gets \emptyset$ \label{line:clear}
    \end{algorithmic}
\end{algorithm}

The barrier $b[\cdot]$ is initialized to $0$ and the dependency lists $B[\cdot]$ to
$\emptyset$; the rest of $\bsdfs$ (\autoref{sec:algorithm}) is unchanged.
The clear of $B[v]$ in line~\ref{line:clear} couples \emph{de-registration} with \emph{clearing}:
a node is removed from a dependency list exactly when the cascade that clears its
barrier passes through. Without it the lists would only grow, accumulating stale
entries; with it they record precisely the nodes whose positive barrier is still in
force, which is what the analysis below exploits.

Gupta and Suzumura~\cite{DBLP:journals/corr/abs-2105-10094} pair lazy registration
with the \emph{tight} (distance-depositing) fruitful case;
but that combination is not always complete, see \cite{boundedcycles2025}.
The lazy scheme is complete precisely because it is coarse: it resets to $0$, so no exact
distance need be maintained, and registration only has to control \emph{which} nodes
are reset.

Next, we show completeness of the lazy scheme.

\paragraph{The registration invariant.}
The mechanism rests on a single edge-local invariant.

\begin{lemma}[Registration Invariant]\label{lem:registration-invariant}
    At every instant of the execution,
    including while a cascade is in progress,
    every node $x$ with $b[x] > 0$ and $x \notin S$
    satisfies $x \in B[w]$ for every successor $w \in \suc(x)$.
\end{lemma}
\begin{proof}
    A barrier is written only by initialization ($0$), by the fruitful reset ($0$), by
    the fruitless assignment ($k - h + 1$), or by $\textsc{Update}$ ($0$).
    Any positive value of $b[x]$ must originate from a fruitless assignment, made by some
    invocation $\sigma = \search(x)$; immediately after it, $\sigma$ inserted $x$ into
    $B[w]$ for every $w \in \suc(x)$ (line~\ref{line:register}). We show $x$ remains in
    all of these lists until the instant in question.

    A node leaves a list only by the clear $B[w] \gets \emptyset$ on line~\ref{line:clear},
    inside some $\textsc{Update}(w)$. Suppose such an $\textsc{Update}(w)$ runs after
    $\sigma$'s assignment and removes $x$, so $x \in B[w]$ when its loop reaches $x$. At
    that moment:
    \begin{itemize}
        \item if $x \notin S$ and $b[x] > 0$, the loop executes $b[x] \gets 0$ ---
              a write to $b[x]$ later than $\sigma$'s, so $\sigma$'s is not the most recent,
              contradicting that $b[x] > 0$ now stems from $\sigma$;
        \item if $x \notin S$ and $b[x] = 0$, the barrier fell from its positive value to
              $0$ after $\sigma$ --- again a later write, the same contradiction;
        \item if $x \in S$, some $\search(x)$ is on the stack; it began after $\sigma$
              returned (since $\sigma$ popped $x$) and either returns before the instant in
              question, writing $b[x]$ later than $\sigma$, or is still on the stack then,
              making $x \in S$, contrary to assumption.
    \end{itemize}
    Every case is contradictory, so no $\textsc{Update}$ removes $x$ from any $B[w]$;
    hence $x \in B[w]$ for every $w \in \suc(x)$.
\end{proof}

\paragraph{Edge-consistency.}
We show the lazy scheme is edge-consistent w.r.t.\ the current search path at every
point at which control is not inside an $\textsc{Update}$ cascade; completeness then
follows from \fullref{lem:pruning-is-permissive} exactly as for the loose scheme. As
before, a cascade transiently breaks edge-consistency and restores it on completion,
so the statement concerns these quiescent points.

The induction over the run is identical to the loose case
(\autoref{sec:loose-scheme}) on entry, fruitless assignment, and push; only the
fruitful cascade differs, and there \autoref{lem:registration-invariant} replaces the
loose closure property. The crux is the following.

\begin{lemma}[Cascade clears predecessors]\label{lem:lazy-osec}
    During $\textsc{Update}(v)$, for every node $y$ on which $\textsc{Update}(y)$ is
    invoked, every off-stack predecessor $x$ of $y$ has $b[x] = 0$ from the moment
    $\textsc{Update}(y)$ returns.
\end{lemma}
\begin{proof}
    Consider the start of $\textsc{Update}(y)$. If $b[x] > 0$ and $x \notin S$ then,
    since $y \in \suc(x)$, \autoref{lem:registration-invariant} gives $x \in B[y]$, so
    the loop of $\textsc{Update}(y)$ reaches $x$ and, as $x \notin S$, sets
    $b[x] \gets 0$. If instead $b[x] = 0$, nothing is required. A cascade only sets
    barriers to $0$ and never raises one, so $b[x] = 0$ persists until
    $\textsc{Update}(y)$ returns.
\end{proof}

With \autoref{lem:lazy-osec}, the fruitful cascade preserves edge-consistency just as the loose reset did.
Consider an off-path edge $(x,y)$ after $\textsc{Update}(v)$ completes.
If $b[y]$ was cleared by the cascade, then $\textsc{Update}(y)$ was invoked.
Since $x$ is an off-stack predecessor of $y$, \autoref{lem:lazy-osec} implies that $b[x]=0$ from the moment $\textsc{Update}(y)$ returns.
Therefore,
$b[x]=0 \le b[y]+1.$
If $b[y]$ was not cleared, then the edge retains the consistency it had before the cascade.
Indeed, $b[x]$ can only have decreased during the cascade, so any inequality
$b[x] \le b[y]+1$
that held before the cascade continues to hold afterward.
The fruitful node itself ends with $b[v]=0$.
Hence, its outgoing edges are consistent for free.
Its incoming edges are covered by \autoref{lem:lazy-osec} with $y=v$, for exactly the same reason as in the loose scheme.
The fruitless case, the push, and the pop are as in the loose scheme.
Thus, the lazy scheme is edge-consistent at every quiescent point of the execution.
By \fullref{lem:pruning-is-permissive}, it follows that the algorithm produces every simple $s$-$t$-path of length $\le k$.

\paragraph{Delay.}
Like the loose scheme, the lazy scheme discards distance information and
resets fruitful nodes to $0$, and it fails on the same family: on the
graphs $(G_k)$ of \autoref{app:loose-breaker} the deferred cascade retraces
the loose reset exactly, so the lazy scheme likewise performs
$\Omega(k^2(n+m))$ steps between two consecutive outputs and admits no 
$O(k(n+m))$ delay bound.

Its main interest is structural. It demonstrates that the fruitful cascade can be driven \emph{lazily},
through dependency lists accumulated during fruitless search and consumed upon success, while still preserving edge-consistency.
More importantly, it shows that correctness depends on the coarse reset itself rather than on the lazy bookkeeping.
Indeed, it is precisely the combination of lazy registration with a \emph{tight} fruitful case that fails
in $\textsc{CYCLE\_SEARCH}$, as shown in \cite{boundedcycles2025}.
The lazy scheme remains complete because it abandons exact distance information and resets fruitful nodes to $0$.


\section{Analyzing the $\textsc{BC-DFS}$ Algorithm}\label{sec:analyzing-bcdfs}

Now that we have established the correctness of $\bsdfs$ and its variants,
we turn to the $\textsc{BC-DFS}$ algorithm of Peng et al.~\cite{10.14778/3372716.3372720, peng_efficient_2021}.
We show that $\textsc{BC-DFS}$ is not complete and identify a corresponding gap in its correctness proof.
The crucial difference between $\textsc{BC-DFS}$ and $\bsdfs$ lies in the handling of fruitful returns.
In $\bsdfs$, a fruitful call always performs an unconditional barrier update through $\fruitful$ (\autoref{alg:fruitful}).
By contrast, $\textsc{BC-DFS}$ propagates a fruitful value only when a guard condition is satisfied.
Its update procedure, reproduced in \autoref{alg:UpdateBarrier}, assigns a value $l$ to a node $u$ only if $b[u] > l$.
Moreover, if the guard fails for the initial call to $\textsc{UpdateBarrier}$, the entire update cascade is skipped.
This behavior appears both in the conference version~\cite{10.14778/3372716.3372720} and in the journal version~\cite{peng_efficient_2021}.
As we show below, suppressing the cascade in this way can leave stale barrier values in place and may cause valid paths to be pruned.

\begin{algorithm}[H]
\caption{\textsc{UpdateBarrier} from $\textsc{BC-DFS}$ due to Peng et al.\label{alg:UpdateBarrier}}
\begin{algorithmic}[1]

\STATE \textbf{$\textsc{UpdateBarrier}(u, l)$}
\STATE \quad\textbf{if} $b[u] > l$  \textbf{then}
\STATE \quad\quad $b[u] \gets l$
\STATE \quad\quad \textbf{for each} $v \in \mathrm{pred}(u)$ \textbf{do}
\STATE \quad\quad\quad \textbf{if} $v \notin S$ \textbf{then}
\STATE \quad\quad\quad\quad $\textsc{UpdateBarrier}(v, l + 1)$

\end{algorithmic}
\end{algorithm}

\subsection{A Counter-Example to $\textsc{BC-DFS}$ Completeness}

Skipping the update cascade can violate edge-consistency and thereby destroy completeness.
This is demonstrated by the graph $X$ depicted in \autoref{fig:bcdfs-counterexample-X}.

\begin{figure}[H]
    \centering
    \begin{tikzpicture}[
        >=latex,
        node distance=18mm,
        every node/.style={draw, circle, minimum width=7mm, minimum height=7mm}
    ]

        \node (A) at (0,2) {A};
        \node (B) at (3,2) {B};
        \node (C) at (0,0) {C};
        \node (D) at (3,0) {D};
        \node (E) at (5,1) {E};

        \draw[->] (A) -- (B);
        \draw[->] (A) -- (C);

        \draw[->, bend left=10] (B) to (C);
        \draw[->, bend left=10] (C) to (B);

        \draw[->, bend left=10] (B) to (D);
        \draw[->, bend left=10] (D) to (B);

        \draw[->] (C) -- (D);

        \draw[->] (B) -- (E);

    \end{tikzpicture}

    \caption{A Counter-Example to $\textsc{BC-DFS}$ Completeness, Graph $X$ ($s=A$, $t=E$, $k=4$).}\label{fig:bcdfs-counterexample-X}
\end{figure}
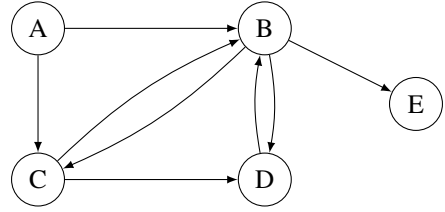

The directed graph $X$ is assumed to be stored in adjacency list format in alphabetical order.
It contains the following simple paths starting at node $A$,
terminating at node $E$ and having length 4 or less:
\begin{samepage}
\begin{enumerate}
  \item $(A,B,E)$ of length 2,
  \item $(A,C,B,E)$ of length 3, and
  \item $(A,C,D,B,E)$ of length 4.
\end{enumerate}
\end{samepage}

However, when $\textsc{BC-DFS}(G,s,t,k)$ with $G=X$, $s=A$, $t=E$, $k=4$ is executed,
it will only find the first two paths and miss the third one.
An execution trace for the counter-example is given in
\autoref{list:bcdfs-counterexample-X}.

While $B$ is on the search path, the fruitless searches of $C$ and $D$
raise their barriers to $b[C] = b[D] = 3$.
When the search of $B$ then returns \emph{fruitful}, $\textsc{BC-DFS}$ skips its $\textsc{UpdateBarrier}$ cascade ---
visible as the absence of any update at \texttt{search(B) fruitful}, the barrier
vector being unchanged across the pop of $B$. The relaxations that would lower
$b[C]$ and $b[D]$ along the edges $(C, B)$ and $(D, B)$ are therefore never
performed, and $b[D] = 3$ persists.

The stale barrier is fatal one step later. When the search path becomes $(A, C)$,
the successor $D$ is examined at depth $h = 1$; the pruning test fails, since
$b[D] + h = 3 + 1 = 4 \ge k$, so $D$ is pruned (trace:
\texttt{D pruned (4 >= k)}). The only $A$-$E$-path through this edge,
$(A, C, D, B, E)$, is thereby lost: $\textsc{BC-DFS}$ outputs $(A, B, E)$ and
$(A, C, B, E)$ but misses the valid path $(A, C, D, B, E)$ proving that the algorithm is not complete.
The same stale $b[D] = 3$ had already pruned $D$ at the deeper path $(A, C, B)$,
where $b[D] + h = 3 + 2 = 5 \ge k$ (trace: \texttt{D pruned (5 >= k)}).

Additional counter-examples were found experimentally by searching small random graphs, of which $X$ is among the smallest.
Moreover, infinitely many counter-examples can be obtained by suitable augmentations of $X$.

\subsection{Proof Analysis of $\textsc{BC-DFS}$ Completeness}

The incompleteness of $\textsc{BC-DFS}$ is not merely an implementation issue.
Upon closer inspection, the correctness proof contains a gap precisely at the point where the algorithm fails on the counter-example above.
The proof of Theorem~1 in~\cite{10.14778/3372716.3372720, peng_efficient_2021} argues that
``for any vertex $u$, $u.\mathrm{bar}$ is correctly maintained in Algorithm $\textsc{BC-DFS}$''.
The argument distinguishes two cases, depending on whether the vertex $x$ coincides with a vertex $v$.
For the case $x \ne v$, the proof explicitly describes how barrier values are propagated along the path $p(u \to x)$.
For the case $x = v$, however, it states only that ``the proof is similar'' and does not provide the corresponding argument.
Our counter-example falls precisely into this omitted case. When $x=v$ and $N=\emptyset$, the propagation argument through the vertices of $N$ becomes vacuous.
At the same time, the guard $u.\mathrm{bar} > l$ in $\textsc{UpdateBarrier}$ may prevent the relaxation from reaching $u$.
Consequently, stale barrier values can remain in place and valid paths may be pruned, exactly as observed in the counter-example.
Thus, the failure exhibited by the counter-example corresponds precisely to the case that is not established in the published proof.

\subsection{Proof Analysis of $\textsc{BC-DFS}$ Time Complexity}

We also identify a gap in the time-complexity argument of $\textsc{BC-DFS}$.
Beyond its relevance to the complexity analysis, the discussion reveals several interesting structural properties of the algorithm.
In~\cite{10.14778/3372716.3372720, peng_efficient_2021}, the central step in the proof of Theorem~2, which states that
``$\textsc{BC-DFS}$ is a polynomial-delay algorithm with $O(km)$ time per output'', is the following argument:
\begin{quote}
Suppose a vertex $u$ is unstacked twice and there is no new output in Algorithm 1.
Let $S_1$ and $S_2$ denote the stack size
after $u$ is pushed into the stack at the first and the second time, respectively.
We have $u.bar = k-S_1+2$ after $u$ is unstacked at the first time.
As there is no new output, the propagation of barrier values will not be invoked.
Thus, $u.bar$ remains the same when $u$ is pushed to stack at the second time.
As $u$ passes the barrier-based pruning in the second visit,
we have $S_2 + u.bar \le k$, and hence $S_2 < S_1$.
\end{quote}

However, this argument does not correctly account for the interleaving of $\search$ and $\fruitful$
(here: $\textsc{UpdateBarrier}$) calls while the recursion unwinds. As a result, the claimed monotonicity property is false:
barrier values may change between two visits of the same vertex even when no new output is produced.
A counter-example demonstrating this behavior is depicted in \autoref{fig:bcdfs-counterexample-Y}.

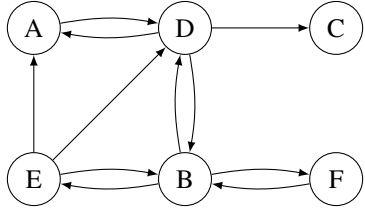
\begin{figure}[H]
    \centering
    \begin{tikzpicture}[
        >=latex,
        every node/.style={draw, circle, minimum width=7mm, minimum height=7mm}
    ]
        \node (A) at (0, 2) {A};
        \node (B) at (2, 0) {B};
        \node (C) at (4, 2) {C};
        \node (D) at (2, 2) {D};
        \node (E) at (0, 0) {E};
        \node (F) at (4, 0) {F};

        \draw[->, bend left=10] (A) to (D);
        \draw[->, bend left=10] (D) to (A);

        \draw[->, bend left=10] (B) to (D);
        \draw[->, bend left=10] (D) to (B);

        \draw[->, bend left=10] (B) to (E);
        \draw[->, bend left=10] (E) to (B);

        \draw[->, bend left=10] (B) to (F);
        \draw[->, bend left=10] (F) to (B);

        \draw[->] (D) -- (C);
        \draw[->] (E) -- (A);
        \draw[->] (E) -- (D);
    \end{tikzpicture}
    \caption{A Counter-Example to $\textsc{BC-DFS}$ Monotonicity, Graph $Y$
             ($s = E$, $t = C$, $k = 6$).}\label{fig:bcdfs-counterexample-Y}
\end{figure}

The graph $Y$ is stored in adjacency-list format in alphabetical order. The complete
execution trace of $\textsc{BC-DFS}(G,s,t,k)$ with $G = Y$, $s = E$, $t = C$, $k = 6$
is given in \autoref{list:bcdfs-counterexample-Y}; the barrier vector
$(b[A], b[B], b[C], b[D], b[E], b[F])$ is shown after every step. $\textsc{BC-DFS}$
produces all three $s$-$t$-paths $(E, A, D, C)$, $(E, B, D, C)$, and $(E, D, C)$,
so completeness is not at issue in this example.

We focus on the node $F$. In the interval between the second output $(E, B, D, C)$ and the
third $(E, D, C)$, two observations are crucial.

First, after the second output $(E, B, D, C)$, the recursion unwinds and the
call $\search(B)$ at search path $(E, B)$ returns fruitful --- its
fruitfulness stems from that earlier output, produced within it at greater depth.
Its $\textsc{UpdateBarrier}$ cascade is therefore invoked during the unwind,
\emph{although no new output has been produced since}. The cascade lowers
$b[B] \gets 2$ and $b[F] \gets 3$. This already contradicts the premise of the
quoted argument --- that in the absence of a new output ``the propagation of
barrier values will not be invoked'': between two consecutive outputs,
cascades of fruitful calls returning on the spine of the previous output do run.

Second, and within the same interval, $F$ is unstacked fruitless \emph{twice}. It is
first searched fruitless at search path $(E, B, F)$, of length $2$, setting
$b[F] \gets 5$; later, after the cascade has reset it to $3$, it is searched fruitless
again at search path $(E, D, B, F)$, of \emph{greater} length $3$, setting
$b[F] \gets 4$. So the two stack sizes at which $F$ is unstacked within one interval
satisfy $S_2 > S_1$, directly contradicting the claimed $S_2 < S_1$. The barrier of $F$
is correspondingly non-monotone, moving $5 \to 3 \to 4$.

Because barrier propagation is invoked without a new output, the step ``$u.bar$ remains
the same when $u$ is pushed to stack the second time'' fails, and with it the conclusion
$S_2 < S_1$. The postulated interval delay bound of $O(km)$ therefore does not follow,
and remains open for $\textsc{BC-DFS}$.
Note that monotonicity does hold in the first interval, where no spine path unwinds and
hence no $\fruitful$ cascade runs.
The trace for $\bsdfs$ on this graph is similar, with additional cascade updates, but the
non-monotone assignments to $b[F]$ are the same. Accounting correctly for the interleaved
$\search$ and $\fruitful$ calls across \emph{all} intervals is the key point of our
$\bsdfs$ delay-bound proof.

\section{Experiments}\label{sec:experiments}

Our experiments address the following questions:
how often \textsc{BC-DFS} misses paths (\autoref{fig:missed-paths}),
what completeness costs, 
in elementary steps (\autoref{fig:steps}, \autoref{fig:steps-accounts})
and in running time (\autoref{fig:runtime}), and
how the measured delays of $\bsdfs$ compare to the constants
proven in \autoref{sec:delay-bounds} (\autoref{fig:delay-bounds}).

All experiments draw from two synthetic graph families; benchmarks on
real-world networks are beyond the scope of this paper, whose subject is
the correctness and the delay of the algorithm rather than its
performance at scale.
$\ER$ denotes directed Erd\H{o}s--R\'enyi graphs 
$G(n,m)$ with $n$ uniform on $\{6,\dots,30\}$ and
mean out-degree $m/n$ log-uniform on $[1,n-1]$, i.e.\ from tree-like
density up to the complete digraph; $\WS$ denotes Watts--Strogatz graphs
with $n = 1000$, degree $6$ and rewiring probability $0.2$, symmetrized to
a digraph with $m = 6000$. In both, $s$ and $t$ are drawn uniformly
without replacement. Each family contributes $1000$ instances per
$k \in \{3,\dots,10\}$, the same instances in every experiment.
Enumeration is capped at $10^{5}$ outputs; instances reaching the cap are
excluded, since a capped run reports no loss by construction. The cap
binds only in $\ER$ and only for $k \ge 5$, excluding $17$, $100$, $167$,
$223$, $262$ and $292$ of $1000$ instances for $k = 5,\dots,10$.

The experiments were executed on an Intel Core i7-14700K 28-core 32 GiB RAM Workstation
running Ubuntu 24.04.4 LTS under WSL2 and using Python 3.12.3.
The algorithms were directly translated from pseudocode to Python,
single threaded and with no further optimizations at the code level.
Timings are the minimum over five passes, the passes separated in time so
that a transient disturbance cannot affect all measurements of one instance.
All source code and logs are publicly available in our GitHub repository~\cite{bsdfs_github}.

\begin{figure}[tbp]
    \centering
    \includegraphics[width=0.8\textwidth]{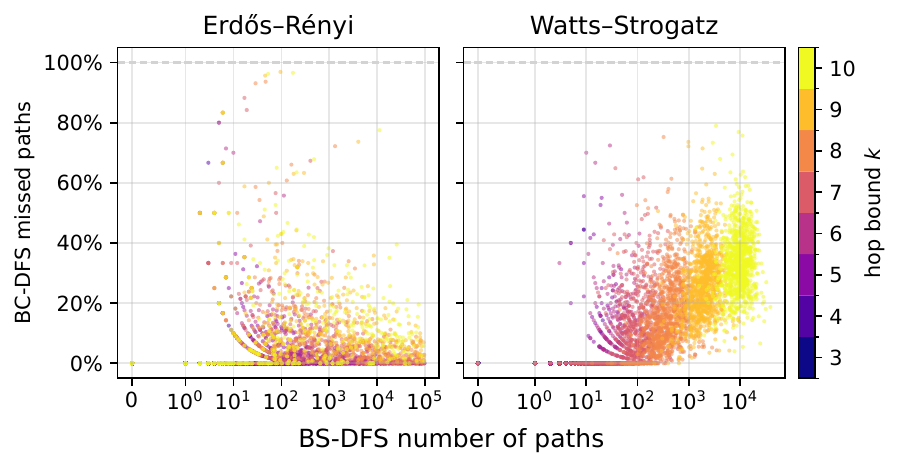}
    \caption{Incompleteness of $\textsc{BC-DFS}$: the fraction of paths it
        fails to report, against the number of paths reported by
        $\textsc{BS-DFS}$ (symmetric-log abscissa, so instances without any path
        appear at $0$). No point lies below the axis: $\textsc{BC-DFS}$ is sound,
        and its output is a subsequence of $\textsc{BS-DFS}$'s, so the loss is
        one-sided. The visible arcs are the level sets of a fixed absolute loss.
        Aggregated over instances, the loss reaches $9.4\%$ on $\ER$ and $33.6\%$
        on $\WS$ at $k = 10$; at that $k$ every $\WS$ instance and $43\%$ of the
        $\ER$ instances lose at least one path.}
    \label{fig:missed-paths}
\end{figure}

\begin{figure}[tbp]
    \centering
    \includegraphics[width=0.8\textwidth]{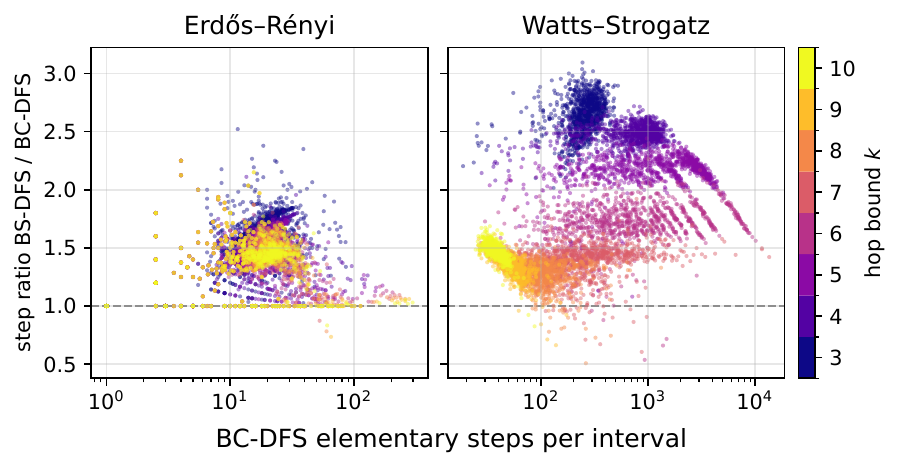}
    \caption{The cost of replacing $\textsc{BC-DFS}$ by $\textsc{BS-DFS}$,
        measured in the elementary steps of \autoref{sec:work-attribution} and
        normalized per interval, so that neither algorithm is charged for output
        the other did not produce. Values above $1$ are the price of
        completeness.
        Aggregated over the instances of one $k$, the ratio lies between
        $1.48$ and $1.78$ on $\ER$ and between $1.28$ and $2.66$ on $\WS$,
        where it is smallest around $k = 7$; individual instances scatter
        more widely.}
\label{fig:steps}
\end{figure}


\begin{figure}[tbp]
    \centering
    \includegraphics[width=0.8\textwidth]{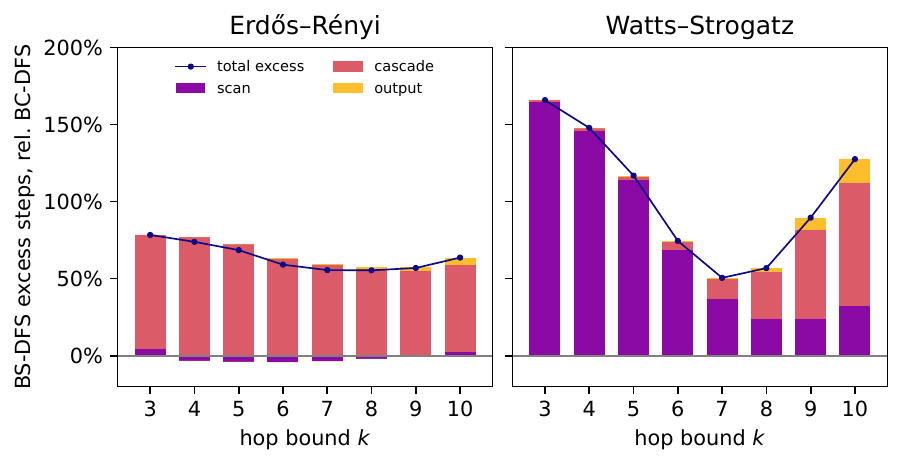}
    \caption{Where the additional steps of $\textsc{BS-DFS}$ are spent: the
        excess in each account of \autoref{sec:work-attribution}, relative to the
        total step count of $\textsc{BC-DFS}$. The three contributions sum to the
        net excess (line). A negative contribution means $\textsc{BC-DFS}$ spends
        more in that account: lacking the origin write of the $\fruitful$
        cascade, its barriers stay weaker, and it admits entries that
        $\textsc{BS-DFS}$ prunes --- on $\ER$ it makes about twice as many
        $\search$ calls. The excess is almost entirely cascade work on $\ER$,
        whereas on $\WS$ it is dominated by the larger search tree, the paths
        missed by $\textsc{BC-DFS}$ being subtrees it never enters.}
    \label{fig:steps-accounts}
\end{figure}

\begin{figure}[tbp]
    \centering
    \includegraphics[width=0.8\textwidth]{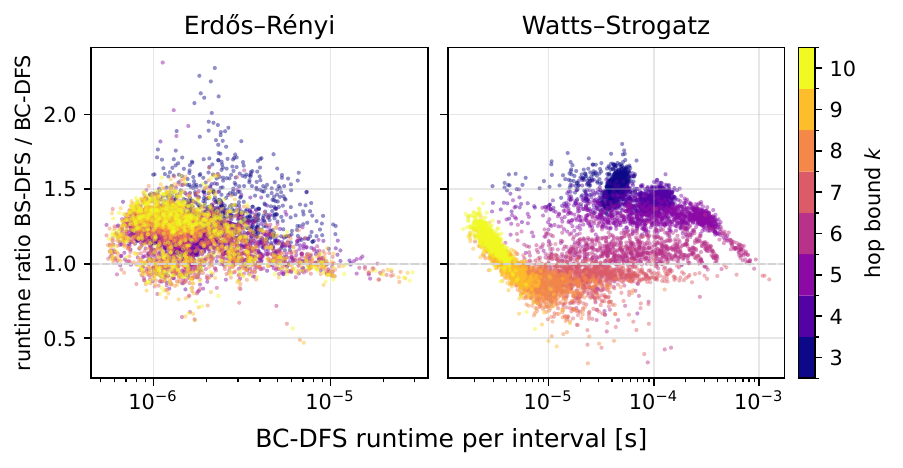}
    \caption{The same comparison as in \autoref{fig:steps}, but in wall-clock time, per interval.
        Timings include a measured harness overhead of about $0.23\,\mu\mathrm{s}$,
        which compresses the ratio towards $1$ at the left edge.
        The agreement in shape with \autoref{fig:steps},
        in particular the minimum near $k = 7$ on $\WS$,
        indicates that the step model of \autoref{sec:work-attribution}
        captures the structure of the observed running time.}
\label{fig:runtime}
\end{figure}

\begin{figure}[tbp]
    \centering
    \includegraphics[width=0.8\textwidth]{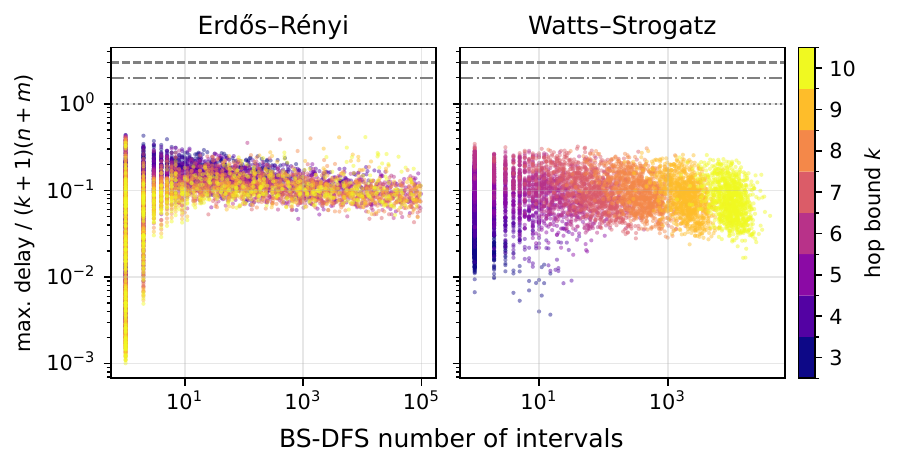}
    \caption{Delay of $\textsc{BS-DFS}$: for each run, the maximum over all
        its intervals of the delay divided by $(k+1)(n+m)$, against the number of
        intervals. The lines mark the proven bounds --- $1$ for the initial
        interval and $3$ for any interval (\fullref{cor:boundary-intervals},
        \autoref{thm:worst-case-delay}), with $2$ for the terminal interval and
        for the amortized delay (\autoref{thm:amortized-delay}). Only the line at
        $3$ bounds the plotted quantity; the observed maximum of $0.44$ on $\ER$
        and $0.34$ on $\WS$ stays below all three, and does not grow over four
        decades of run length. The instances at abscissa $1$ produce no output;
        their single interval is the fruitless case of
        \fullref{cor:boundary-intervals}.}
    \label{fig:delay-bounds}
\end{figure}

\FloatBarrier

Taken together, the experiments show that the defect of $\textsc{BC-DFS}$
is substantial  --- a third of all paths on $\WS$ at $k = 10$,
with every instance losing paths --- and that repairing it is affordable:
$\bsdfs$ needs between $1.3$ and $2.7$ times as many elementary steps
per interval, and at most $1.6$ times the running time.
On $\WS$ at $k = 7$ and $k = 8$ it is the faster of the two.
Its delay never exceeds $0.44\,(k+1)(n+m)$, less than a sixth of the
bound proven in \fullref{thm:worst-case-delay}, and does not grow with
the length of the run.

\section{Conclusions and Outlook}\label{sec:conclusions-and-outlook}

We have presented $\bsdfs$, a bounded-scope depth-first search that enumerates all
simple $s$-$t$-paths --- and, via node splitting, all simple $s$-cycles ---
of length at most $k$ in a directed graph, having a delay of $O(k(n+m))$.
The hidden constants are small: at most $3(k+1)(n+m)$ elementary steps
between consecutive events (\autoref{thm:worst-case-delay}) and, amortized,
at most $2(k+1)(n+m)$ steps per event, the $p$-th event being reached
within $2p(k+1)(n+m)$ steps (\autoref{thm:amortized-delay}).

The algorithm follows the line of barrier-based methods initiated by Johnson and
refined by Peng et al.\ and by Gupta and Suzumura, but it is built around a
single explicit invariant, \emph{edge-consistency}:
the barrier of an off-path node never exceeds one plus the barrier of any off-path successor.
Read as a heuristic in the sense of $A^*$, edge-consistency is exactly admissibility maintained as a
search-time invariant, and it yields a transparent completeness argument ---
an admissible barrier never prunes a path that respects the budget.

The same invariant organizes a whole family of schemes that differ only in how
barriers are revised (\autoref{sec:bsdfs-variants}).
The \emph{tight} scheme deposits exact distances and pays for its delay bound with a
correctness proof in which completeness, the exact return value, and
preservation of edge-consistency are mutually entangled. At the other end, the
\emph{loose} (\autoref{sec:loose-scheme}) and \emph{lazy}
(\autoref{sec:lazy-scheme}) schemes discard distance information by resetting
fruitful nodes to zero; this trades barrier precision for a markedly simpler
correctness argument, since zero is the most admissible value and no
exact-value claim is needed. The lazy scheme additionally shows that the
backward cascade can be driven by dependency lists accumulated during fruitless
search and consumed on success, while preserving edge-consistency.
The $O(k(n+m))$ delay bound, however, is provably lost: on the family of \autoref{app:loose-breaker},
both the loose and the lazy scheme perform $\Omega(k^2(n+m))$ steps between
two consecutive outputs (\autoref{thm:loose-breaker}).

The framework also makes precise where a published algorithm goes wrong.
In $\textsc{BC-DFS}$ (\autoref{sec:analyzing-bcdfs}) a single misplaced guard
causes correct barrier relaxations to be silently dropped, breaking
edge-consistency and, with it, completeness; our counter-example and the
matching gap in the original proof both sit exactly at this point.
The accompanying experiments (\autoref{sec:experiments}) suggest that the resulting
missed-path phenomenon is not confined to isolated counter-examples:
under the graph models considered, the fraction of missed paths 
reaches $9.4\%$ in one model and $33.6\%$ in the other at $k = 10$.

As already proposed for many other path and cycle search algorithms~\cite{
    DBLP:journals/siamcomp/Johnson75,
    DBLP:journals/corr/abs-2105-10094,
    10.14778/3372716.3372720,
    peng_efficient_2021,
    DBLP:conf/spaa/BlanusaIA22,
    DBLP:journals/topc/BlanusaAI23},
our  $\bsdfs$ algorithm can be combined with preprocessing steps like
strongly connected component decomposition and breadth-first search
for reachability from one or both sides (bidirectional search)
to further reduce the search space and improve practical performance.
One practical optimization is to initialize the barrier values in $\bsdfs$ by
a breadth-first search (BFS) from the target node $t$ in the reverse graph,
which sets the initial barrier value of each node to its exact distance towards $t$.
Hence, the initialization is admissible and edge-consistent, and $\bsdfs$ preserves
edge-consistency thereafter, so completeness is preserved.
This initialization allows for more pruning of the search space and thus practically
faster execution, especially if the nodes $x$ which cannot reach $t$ in budget
are initialized with $b[x] = k + 1$ thus being eliminated from the search.

In this paper, edge-consistency is stated for unit edge weights and a single
source and target; extending it to multiple targets and to non-simple and
undirected graphs appears to be a natural direction, and one for which the
invariant-based proof technique developed here should carry over. Weighted
budgets are a different matter. Edge-consistency is a metric condition and
generalizes verbatim to nonnegative weights, so we expect completeness to
survive, with the cascade's unit-layer propagation replaced by Dijkstra-like
label-setting. The delay analysis, however, rests on the unit value space:
barriers take at most $k+1$ values, and the monotone write chains of
\autoref{sec:delay-bounds} are short only for that reason. Integer weights
with budget $K$ yield only a pseudo-polynomial $O(K(n+m))$ by the same
counting, and real weights no bound at all. Polynomial delay as such remains
available by answering every pruning test with a fresh shortest-path
computation; whether the barrier certificates admit comparable \emph{reuse}
under weights is open.

A related question is how far edge-consistency can be relaxed:
completeness needs only local admissibility (when tested for),
which global edge-consistency implies but strictly exceeds,
so characterizing the full class of admissible --- and hence
complete --- bounded-scope schemes, edge-consistent or not, is open.

Our experiments are deliberately limited in scope: they use two classes of
random graph models and limited $k$ values to evaluate $\bsdfs$.
A broader empirical and comparative study measuring per-output delay
on real-world road, social, and similar networks would require substantial code-level optimization.
Such engineering work lies outside the focus of this paper and is left to future work.

\section*{Acknowledgements}
All ideas, claims, and proofs, the $\bsdfs$ algorithm and its variants,
and the $\textsc{BC-DFS}$ counter-examples and findings are the authors' own work.
Generative AI tools 
were used under the authors' direction for proofreading, checking,
polishing of the author-written text, and for minor coding assistance.
The graph family $(G_k)$ of \autoref{app:loose-breaker} was discovered
with assistance from Claude (Anthropic), model Claude Fable~5.
All AI-assisted material was reviewed and verified by the authors,
who take full responsibility for the content of the paper.

\hrulefill

\appendix

\section{Execution Trace of Counter-Example Graph $X$}\label{sec:counter-ex-X}

    \begin{lstlisting}[
        caption={Execution Trace of $\textsc{BC-DFS}(G,s,t,k)$ with $G=X$, $s=A$, $t=E$, $k=4$},
        label=list:bcdfs-counterexample-X
    ]
S path: barriers  : trace
      : A B C D E :
      : 0 0 0 0 0 : search(A) enter
A     : 0 0 0 0 0 : .search(B) enter
AB    : 0 0 0 0 0 : ..search(C) enter
ABC   : 0 0 0 0 0 : ...search(B) B pruned (in S)
ABC   : 0 0 0 0 0 : ...search(D) enter
ABCD  : 0 0 0 0 0 : ....search(B) B pruned (in S)
ABCD  : 0 0 0 0 0 : ....search(D) fruitless bar[D] := 2 (was 0)
ABC   : 0 0 0 2 0 : ...search(D) exit
ABC   : 0 0 0 2 0 : ...search(C) fruitless bar[C] := 3 (was 0)
AB    : 0 0 3 2 0 : ..search(C) exit
AB    : 0 0 3 2 0 : ..search(D) enter
ABD   : 0 0 3 2 0 : ...search(B) B pruned (in S)
ABD   : 0 0 3 2 0 : ...search(D) fruitless bar[D] := 3 (was 2)
AB    : 0 0 3 3 0 : ..search(D) exit
AB    : 0 0 3 3 0 : ..search(E) enter
ABE   : 0 0 3 3 0 : ...search(E) output ['A', 'B', 'E']
AB    : 0 0 3 3 0 : ..search(B) fruitful
A     : 0 0 3 3 0 : .search(B) exit
A     : 0 0 3 3 0 : .search(C) enter
AC    : 0 0 3 3 0 : ..search(B) enter
ACB   : 0 0 3 3 0 : ...search(C) C pruned (in S)
ACB   : 0 0 3 3 0 : ...search(D) D pruned (5 >= k)
ACB   : 0 0 3 3 0 : ...search(E) enter
ACBE  : 0 0 3 3 0 : ....search(E) output ['A', 'C', 'B', 'E']
ACB   : 0 0 3 3 0 : ...search(B) fruitful
AC    : 0 0 3 3 0 : ..search(B) exit
AC    : 0 0 3 3 0 : ..search(D) D pruned (4 >= k)
AC    : 0 0 3 3 0 : ..search(C) fruitful
AC    : 0 0 3 3 0 : ..update(C) bar[C] := 2 (was 3)
A     : 0 0 2 3 0 : .search(C) exit
A     : 0 0 2 3 0 : .search(A) fruitful
      : 0 0 2 3 0 : search(A) exit
    \end{lstlisting}

\section{Execution Trace of Counter-Example Graph $Y$}\label{sec:counter-ex-Y}

    \begin{lstlisting}[
        caption={Execution Trace of $\textsc{BC-DFS}(G,s,t,k)$ with $G=Y$, $s=E$, $t=C$, $k=6$},
        label=list:bcdfs-counterexample-Y
    ]
S path :   barriers  : trace
       : A B C D E F :
       : 0 0 0 0 0 0 : search(E) enter
E      : 0 0 0 0 0 0 : .search(A) enter
EA     : 0 0 0 0 0 0 : ..search(D) enter
EAD    : 0 0 0 0 0 0 : ...search(A) A pruned (in S)
EAD    : 0 0 0 0 0 0 : ...search(B) enter
EADB   : 0 0 0 0 0 0 : ....search(D) D pruned (in S)
EADB   : 0 0 0 0 0 0 : ....search(E) E pruned (in S)
EADB   : 0 0 0 0 0 0 : ....search(F) enter
EADBF  : 0 0 0 0 0 0 : .....search(B) B pruned (in S)
EADBF  : 0 0 0 0 0 0 : .....search(F) fruitless bar[F] := 3 (was 0)
EADB   : 0 0 0 0 0 3 : ....search(F) exit
EADB   : 0 0 0 0 0 3 : ....search(B) fruitless bar[B] := 4 (was 0)
EAD    : 0 4 0 0 0 3 : ...search(B) exit
EAD    : 0 4 0 0 0 3 : ...search(C) enter
EADC   : 0 4 0 0 0 3 : ....search(C) output ['E', 'A', 'D', 'C']
EAD    : 0 4 0 0 0 3 : ...search(D) fruitful
EA     : 0 4 0 0 0 3 : ..search(D) exit
EA     : 0 4 0 0 0 3 : ..search(A) fruitful
E      : 0 4 0 0 0 3 : .search(A) exit
E      : 0 4 0 0 0 3 : .search(B) enter
EB     : 0 4 0 0 0 3 : ..search(D) enter
EBD    : 0 4 0 0 0 3 : ...search(A) enter
EBDA   : 0 4 0 0 0 3 : ....search(D) D pruned (in S)
EBDA   : 0 4 0 0 0 3 : ....search(A) fruitless bar[A] := 4 (was 0)
EBD    : 4 4 0 0 0 3 : ...search(A) exit
EBD    : 4 4 0 0 0 3 : ...search(B) B pruned (in S)
EBD    : 4 4 0 0 0 3 : ...search(C) enter
EBDC   : 4 4 0 0 0 3 : ....search(C) output ['E', 'B', 'D', 'C']
EBD    : 4 4 0 0 0 3 : ...search(D) fruitful
EB     : 4 4 0 0 0 3 : ..search(D) exit
EB     : 4 4 0 0 0 3 : ..search(E) E pruned (in S)
EB     : 4 4 0 0 0 3 : ..search(F) enter
EBF    : 4 4 0 0 0 3 : ...search(B) B pruned (in S)
EBF    : 4 4 0 0 0 3 : ...search(F) fruitless bar[F] := 5 (was 3)
EB     : 4 4 0 0 0 5 : ..search(F) exit
EB     : 4 4 0 0 0 5 : ..search(B) fruitful
EB     : 4 4 0 0 0 5 : ..update(B) bar[B] := 2 (was 4)
EB     : 4 2 0 0 0 5 : ..update(F) bar[F] := 3 (was 5)
E      : 4 2 0 0 0 3 : .search(B) exit
E      : 4 2 0 0 0 3 : .search(D) enter
ED     : 4 2 0 0 0 3 : ..search(A) enter
EDA    : 4 2 0 0 0 3 : ...search(D) D pruned (in S)
EDA    : 4 2 0 0 0 3 : ...search(A) fruitless bar[A] := 5 (was 4)
ED     : 5 2 0 0 0 3 : ..search(A) exit
ED     : 5 2 0 0 0 3 : ..search(B) enter
EDB    : 5 2 0 0 0 3 : ...search(D) D pruned (in S)
EDB    : 5 2 0 0 0 3 : ...search(E) E pruned (in S)
EDB    : 5 2 0 0 0 3 : ...search(F) enter
EDBF   : 5 2 0 0 0 3 : ....search(B) B pruned (in S)
EDBF   : 5 2 0 0 0 3 : ....search(F) fruitless bar[F] := 4 (was 3)
EDB    : 5 2 0 0 0 4 : ...search(F) exit
EDB    : 5 2 0 0 0 4 : ...search(B) fruitless bar[B] := 5 (was 2)
ED     : 5 5 0 0 0 4 : ..search(B) exit
ED     : 5 5 0 0 0 4 : ..search(C) enter
EDC    : 5 5 0 0 0 4 : ...search(C) output ['E', 'D', 'C']
ED     : 5 5 0 0 0 4 : ..search(D) fruitful
E      : 5 5 0 0 0 4 : .search(D) exit
E      : 5 5 0 0 0 4 : .search(E) fruitful
       : 5 5 0 0 0 4 : search(E) exit
    \end{lstlisting}

\section{A Graph Family Breaking the Loose and Lazy Delay}
\label{app:loose-breaker}

This appendix exhibits a graph family $(G_k)$ on which the loose scheme
(\autoref{sec:loose-scheme}) and the lazy scheme (\autoref{sec:lazy-scheme})
perform $\Omega(k^2(n+m))$ steps between two consecutive outputs.
Hence, neither scheme attains the $O(k(n+m))$ delay bound that
\fullref{thm:worst-case-delay} establishes for the tight scheme: the distance
information deposited by the fruitful cascade is essential to the delay
bound, not only to its proof. We give the construction and a counting
argument with generous constants; the demonstration code, exact hub-entry
counts, and successor-scan measurements are available at~\cite{bsdfs_github}.

\paragraph{Construction.}
Let $k \equiv 0 \pmod 4$, $k \ge 8$, and set $p = k/2 - 1$, $D = 2p$, and
$F = k$. The node set of $G_k$ consists of $s$, $t$, a \emph{spine}
$v_1, \dots, v_p$, a \emph{chain} $d_1, \dots, d_D$, a \emph{hub} $x$, and a
\emph{fan} $f_1, \dots, f_F$ of sinks. The successors of each node, in scan
order, are:
\begin{center}
\begin{tabular}{ll}
node & successors, in scan order \\[1pt]
\hline \\[-9pt]
$s$              & $v_1,\; t$ \\
$v_i$ \ ($i<p$)  & $v_{i+1},\; d_1$ \\
$v_p$            & $t,\; d_1$ \\
$d_j$ \ ($j<D$)  & $d_{j+1},\; x,\; [\,v_i \text{ if } j = j(i)\,]$ \\
$d_D$            & $x,\; [\,v_p\,]$ \\
$x$              & $f_1,\dots,f_F,\; d_1$ \\
$f_q$            & --- \\
\end{tabular}
\end{center}
where the \emph{branch} edges $(d_{j(i)}, v_i)$ with $j(i) = p + i$ exist
for $i = 2, \dots, p$. The sizes are $n = 5k/2$ and $m = 9k/2 - 6$, so
$n + m = 7k - 6 = \Theta(k)$. \autoref{fig:loose-breaker} shows $G_8$, the smallest member.

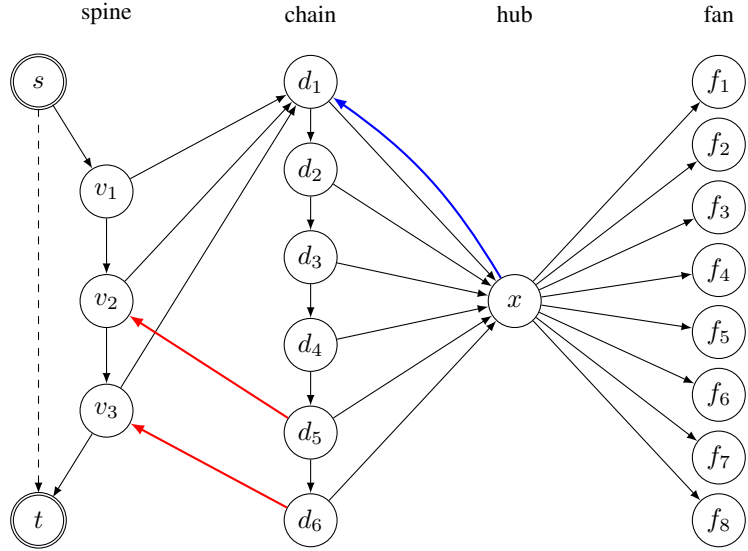
\begin{figure}[H]
    \centering
    \begin{tikzpicture}[
        >=latex,
        every node/.style={draw, circle, minimum size=7mm, inner sep=1pt},
        heading/.style={draw=none, font=\small},
        branch/.style={->, red, thick},
        conduit/.style={->, blue, thick}
    ]
        \node[heading] at (1.8, 0.9) {spine};
        \node[heading] at (4.5, 0.9) {chain};
        \node[heading] at (7.2, 0.9) {hub};
        \node[heading] at (9.9, 0.9) {fan};
        \node (s)  [double] at (0.9,  0.0)  {$s$};
        \node (v1) at (1.8, -1.45) {$v_1$};
        \node (v2) at (1.8, -2.9)  {$v_2$};
        \node (v3) at (1.8, -4.35) {$v_3$};
        \node (t)  [double] at (0.9, -5.8)  {$t$};
        \node (d1) at (4.5,  0.0)  {$d_1$};
        \node (d2) at (4.5, -1.16) {$d_2$};
        \node (d3) at (4.5, -2.32) {$d_3$};
        \node (d4) at (4.5, -3.48) {$d_4$};
        \node (d5) at (4.5, -4.64) {$d_5$};
        \node (d6) at (4.5, -5.8)  {$d_6$};
        \node (x)  at (7.2, -2.9)  {$x$};
        \node (f1) at (9.9,  0.00) {$f_1$};
        \node (f2) at (9.9, -0.83) {$f_2$};
        \node (f3) at (9.9, -1.66) {$f_3$};
        \node (f4) at (9.9, -2.49) {$f_4$};
        \node (f5) at (9.9, -3.31) {$f_5$};
        \node (f6) at (9.9, -4.14) {$f_6$};
        \node (f7) at (9.9, -4.97) {$f_7$};
        \node (f8) at (9.9, -5.80) {$f_8$};
        \draw[->] (s)  -- (v1);
        \draw[->, dashed] (s) -- (t);
        \draw[->] (v1) -- (v2);
        \draw[->] (v2) -- (v3);
        \draw[->] (v3) -- (t);
        \draw[->] (v1) -- (d1);
        \draw[->] (v2) -- (d1);
        \draw[->] (v3) -- (d1);
        \draw[->] (d1) -- (d2);
        \draw[->] (d2) -- (d3);
        \draw[->] (d3) -- (d4);
        \draw[->] (d4) -- (d5);
        \draw[->] (d5) -- (d6);
        \draw[->] (d1) -- (x);
        \draw[->] (d2) -- (x);
        \draw[->] (d3) -- (x);
        \draw[->] (d4) -- (x);
        \draw[->] (d5) -- (x);
        \draw[->] (d6) -- (x);
        \draw[branch] (d5) -- (v2);
        \draw[branch] (d6) -- (v3);
        \draw[->] (x) -- (f1);
        \draw[->] (x) -- (f2);
        \draw[->] (x) -- (f3);
        \draw[->] (x) -- (f4);
        \draw[->] (x) -- (f5);
        \draw[->] (x) -- (f6);
        \draw[->] (x) -- (f7);
        \draw[->] (x) -- (f8);
        \draw[conduit] (x) to[bend right=12] (d1);
    \end{tikzpicture}
    \caption{The graph $G_8$ with branch edges $(d_{j(i)}, v_i)$ (red), the
    conduit $(x, d_1)$ (blue), and the edge $(s,t)$ scanned last (dashed).}
    \label{fig:loose-breaker}
\end{figure}

\paragraph{Output structure.}
Exactly two simple $s$-$t$-paths of length $\le k$ exist:
$P_1 = s, v_1, \dots, v_p, t$ of length $p + 1 = k/2$, and $P_2 = s, t$.
Indeed, any other path enters the chain from some spine node $v_r$, i.e.,
$d_1$ is reached at depth $r + 1$ and $d_j$ at depth $r + j$; it can return
to the spine only via a branch, reaching $v_i$ at depth $r + j(i) + 1$ and
$v_p$ at depth $r + j(i) + 1 + (p - i)$, so $t$ is scanned at depth
\[
h \;=\; r + j(i) + 1 + (p - i) \;=\; r + 2p + 1 \;=\; k + r - 1 \;\ge\; k ,
\]
independently of $i$ since $j(i) = p + i$; the scan guard $b[t] + h < k$
fails. The hub admits no return at all: its successors are the fan sinks and
the conduit $(x, d_1)$, and $d_1$ is on the search path whenever $x$ is
entered. Since both schemes emit outputs in lexicographic path order
(\autoref{sec:loose-scheme}) and $(s,t)$ is scanned last, the entire
adversarial work below falls into the single interval between the two
outputs.

\paragraph{Interval anatomy.}
Between the two outputs the search backtracks along the spine. A call is
fruitful if and only if its subtree emitted an output, so within the
interval exactly the spine calls $v_p, \dots, v_1$ pop fruitful (each
contains output $P_1$), and every chain, hub, and fan call pops fruitless.
Before each fruitful pop, $v_i$ scans its last successor $d_1$, launching a
fruitless exploration $\expll{i}$ of the chain--hub--fan region rooted at
depth $i$.

\paragraph{Admission invariant.}
Within the interval, a barrier of a chain node $d_j$ is either $0$
(initialization or a reset) or was written as $k - (i' + j) + 1$ by the
fruitless pop of $d_j$ at depth $i' + j$ inside an earlier exploration
$\expll{i'}$. Since the roots $i$ strictly decrease over the interval,
$i' > i$ holds whenever $\expll{i}$ scans $d_j$, and the scan guard, tested
at the parent depth $i + j - 1$, evaluates to
\[
\bigl(k - (i' + j) + 1\bigr) + (i + j - 1) \;=\; k - (i' - i) \;\le\; k - 1
\;<\; k ,
\]
and to $i + j - 1 < k$ in the zero case. Hence, $\expll{i}$ enters at least
$d_1, \dots, d_{k - i - 1}$, and each of these pops fruitless at depth
$i + j \le k - 1$, writing the positive barrier $k - (i + j) + 1 \ge 2$.

\paragraph{Counting.}
We now count, with generous constants, the work of the loose scheme in this
single interval.

\begin{enumerate}
    \item \emph{$\Omega(k)$ arming resets.} Let $2 \le i \le k/4$.
        When $v_i$ pops, the exploration $\expll{i}$ has just written
        positive barriers on $d_1, \dots, d_{j(i)}$ (note $j(i) = k/2 - 1 + i \le
        k - i - 1$ for $i \le k/4$), and no reset ran in between. The backward
        cascade $\textsc{Reset}(v_i)$ therefore proceeds from the branch
        predecessor $d_{j(i)}$ down the chain to $d_1$ and through the conduit
        predecessor to the hub: after $\textsc{Reset}(v_i)$ completes, $b[x] = 0$,
        and since every chain node is a predecessor of $x$, also $b[d_j] = 0$ for
        all $j$.
    \item \emph{$\Omega(k)$ hub entries per re-armed exploration.} The next
        exploration $\expll{i-1}$, with root $r = i - 1 \le k/4  -
        1$, starts on an all-zero chain and hub. It descends the chain and, on
        backtrack, each $d_j$ scans $x$ at parent depth $r + j$, deepest first.
        The first scan meets $b[x] = 0$; each entry of $x$ at depth $r + j + 1$
        pops fruitless and writes $b[x] = k - r - j$, so the next scan, from
        $d_{j-1}$, evaluates the guard to $(k - r - j) + (r + j - 1) = k - 1 < k$
        and is admitted again. Hence, $x$ is entered at least $k - r - 2 \ge k/2$
        times within $\expll{i-1}$.
    \item \emph{$\Theta(k)$ scans per hub entry.} Every entry of $x$ runs the
        successor loop over all $F + 1 = k + 1$ edges, unconditionally of whether
        any fan node is entered.
\end{enumerate}
The explorations $\expll{i-1}$ for $i = 2, \dots, k/4$ are
pairwise distinct, so the interval contains at least
\[
\bigl(k/4  - 1\bigr) \cdot \frac{k}{2} \cdot (k + 1)
\;=\; \Omega(k^3) \;=\; \Omega\bigl(k^2 (n+m)\bigr)
\]
successor scans.

\paragraph{The lazy scheme.}
The lazy scheme breaks on the same family with the same counting. On $G_k$,
the registrations coincide with the nonzero backward relation used above:
the fruitless pops of $\expll{i}$ register $d_{j-1} \in B[d_j]$,
$x \in B[d_1]$ (conduit), and $d_{j(i)} \in B[v_i]$ (branch), so
$\textsc{Update}(v_i)$ retraces $\textsc{Reset}(v_i)$ and re-arms the hub
under the same condition; steps 2 and 3 are unchanged. Deferral buys
nothing here: the $F$ registrations performed at each fruitless pop of $x$
are re-performed in every re-armed exploration and cannot be amortized.

\paragraph{The tight scheme.}
The tight scheme is immune, structurally and unconditionally. Structurally,
the calibrated cascade from $v_i$ deposits values $d + 1$ that grow along
the chain, so the value offered to $x$ is large and the guard $b > d + 1$
of \fullref{alg:fruitful} stops the propagation early --- the hub is never
re-armed, and each exploration after the first admits at most one hub entry
(via $d_1$). Unconditionally, \fullref{thm:worst-case-delay} caps the interval
at $3(k+1)(n+m)$ steps regardless.

\medskip
We summarize the results in
\begin{theorem}[Loose/Lazy Lower Bound]\label{thm:loose-breaker}
On the family $(G_k)$, both the loose and the lazy scheme perform
$\Omega(k^2(n+m))$ successor scans (hence steps) between the first and the
second output. Consequently, neither scheme admits an $O(k(n+m))$ delay bound.
\end{theorem}

\bibliography{bsdfs.bib}
\bibliographystyle{abbrvurl}

\end{document}